\documentclass[cmp,numbook,envcountsame]{svjour}
\pdfoutput=1
\usepackage{amsmath}
\usepackage{amsfonts}
\usepackage{amssymb}
\usepackage{graphicx} 
\providecommand{\U}[1]{\protect\rule{.1in}{.1in}}

\begin{document}

\title{Three proofs of the Makeenko--Migdal equation for Yang--Mills theory on the plane}
\author{
Bruce K. Driver\inst{1} \and
Brian C. Hall\inst{2} \thanks{Supported in part by NSF grant DMS-1301534} \and
Todd Kemp\inst{3} \thanks{Supported in part by NSF CAREER award DMS-1254807}}
\institute{University of California San Diego, Department of Mathematics, La Jolla, CA
92093 USA, \email{bdriver@ucsd.edu} \and
University of Notre Dame, Department of Mathematics, Notre Dame, IN 46556 USA,
\email{bhall@nd.edu} \and
University of California San Diego, Department of Mathematics, La Jolla, CA
92093 USA, \email{tkemp@math.ucsd.edu}}

\maketitle

\begin{abstract}
We give three short proofs of the Makeenko--Migdal equation for the
Yang--Mills measure on the plane, two using the edge variables and one using
the loop or lasso variables. Our proofs are significantly simpler than the
earlier pioneering rigorous proofs given by T. L\'{e}vy and by A. Dahlqvist.
In particular, our proofs are \textquotedblleft local\textquotedblright\ in
nature, in that they involve only derivatives with respect to variables
adjacent to the crossing in question. In an accompanying paper with F.
Gabriel, we show that two of our proofs can be adapted to the case of
Yang--Mills theory on any compact surface.

\end{abstract}



\section{Introduction}

The (Euclidean) Yang--Mills field theory describes a random connection on a
principal bundle for a compact Lie group $K,$ known as the structure group. In
two dimensions, the theory is tractable and has been studied extensively. In
particular, for Yang--Mills theory on the plane, it is possible to use a gauge
fixing to make the measure Gaussian, opening the door to rigorous calculations
in a continuum setting. This approach was developed simultaneously in two
papers: \cite{GKS} by L. Gross, C. King, and A. Sengupta; and \cite{Dr} by B. Driver.

The typical objects of study in the theory are the Wilson loop functionals,
given by%
\begin{equation}
\mathbb{E}\{\mathrm{trace}(\mathrm{hol}(L))\}, \label{wilsonLoop}%
\end{equation}
where $\mathbb{E}$ denotes the expectation value with respect to the
Yang--Mills measure, $\mathrm{hol}(L)$ denotes the holonomy of the random
connection around a loop $L,$ and the trace is taken in some fixed
representation of the structure group $K.$ If $L$ is traced out on a graph in
the plane, work of Driver \cite[Theorem 6.4]{Dr} gives a formula for the
Wilson loop functional in terms of the \textit{heat kernel measure} on $K.$
(See (\ref{DriversForm}) below.) One noteworthy feature of the two-dimensional
Yang--Mills theory is its invariance under area-preserving diffeomorphisms.
This invariance is reflected in Driver's formula: the expectation
\eqref{wilsonLoop} may be expressed as a function (determined by the topology
of the graph) of all the areas of the faces of the graph.

The \textit{Makeenko--Migdal equation} relates variations of a Wilson loop
functional in the neighborhood of a simple crossing to the associated Wilson
loops on either side of the crossing, in the case $K=U(N)$. The original
equations, in any dimension, were the subject of \cite{MM}. In \cite[Section
4]{KK}, V. A. Kazakov and I. K. Kostov show that in two dimensions, the
\textquotedblleft keyboard-type\textquotedblright\ variation in Eq. (3) of
\cite{MM} can be interpreted as the alternating sum of derivatives of the
Wilson loop functional with respect to the areas of the faces surrounding a
simple crossing. (See \cite[Equation 24]{KK}, \cite[Equation 9]{KazakovU(N)},
and \cite[Equation 6.4]{GG}.) L\'{e}vy \cite{Levy} then provided a rigorous
proof of the planar Makeenko--Migdal equation, using Driver's formula. A
different proof was subsequently given by A. Dahlqvist in \cite{Dahl}. In this
paper, we offer three new, short proofs of the equation. As we show in an
accompanying paper with F. Gabriel \cite{DHKsurf}, two of these proofs can be
adapted to give a new result, namely a rigorous proof of the Makeenko--Migdal
equation for Yang--Mills theory on an arbitrary compact surface.

We use the bi-invariant metric on $U(N)$ whose value on the Lie algebra
$\mathfrak{u}(N)=T_{e}(U(N))$ is a scaled version of the Hilbert--Schmidt
inner product:%
\begin{equation}
\left\langle X,Y\right\rangle =N\mathrm{trace}(X^{\ast}Y). \label{HSN}%
\end{equation}
It is then convenient to express the Wilson loop functionals using the
\textit{normalized} trace,%
\[
\mathrm{tr}(A):=\frac{1}{N}\mathrm{trace}(A).
\]

Suppose now that $L$ is a loop that is traced out on an oriented graph in the plane. We
now explain what it means for $L$ to have a simple crossing at a vertex $v.$
First, we assume that the graph has four edges incident to $v,$ where we
count an edge~$e$ twice if both the initial and final vertices of $e$ are
equal to $v.$ Second, we assume that $L,$ when viewed as a map of the circle
into the plane, passes through $v$ exactly twice. Third, we assume that each
time $L$ passes through $v,$ it comes in along one edge and passes
\textquotedblleft straight across\textquotedblright\ to the cyclically
opposite edge. Last, we assume that $L$ traverses two of the edges on one
pass through $v$ and the remaining two edges on the other pass through $v.$ 

If $L$ has a simple crossing at $v$, we may parametrize $L$ with time-interval $[0,1]$
and with $L(0)=L(1)=v,$ so
that there is a unique $s_{0}\in (0,1)$ with $L(s_{0})=v.$ Let us label the 
\textit{outgoing} edges $e_{1},\ldots ,e_{4}$ at $v$ in cyclic order, with
the first edge traversed by $L$ labeled as $e_{1}.$ The last edge traversed
by $L$ will then be $e_{3}^{-1}.$ We may choose the labeling of the
remaining two edges so that the first return to $v$ is along $e_{4}^{-1}.$
Thus, $L$ initially leaves $v$ along $e_{1},$ returns along $e_{4}^{-1},$
leaves again along $e_{2},$ and finally returns again along $e_{3}^{-1}.$
Note that depending on the curve $L,$ the cyclic ordering of the edges may
be either clockwise or counter-clockwise. Having labeled the edges in cyclic
order, we then label the faces $F_{1},\ldots ,F_{4}$ adjacent to $v$ in
cyclic order so that $e_{1}$ lies between $F_{4}$ and $F_{1},$ $e_{2}$ lies
between $F_{1}$ and $F_{2},$ etc. See Figure \ref{mmplot.fig} for an example, where $L$
initially leaves $v$ along the edge between $F_{1}$ and $F_{4}.$

\begin{figure}[ptb]
\centering
\includegraphics[
height=2.1949in,
width=2.1949in
]{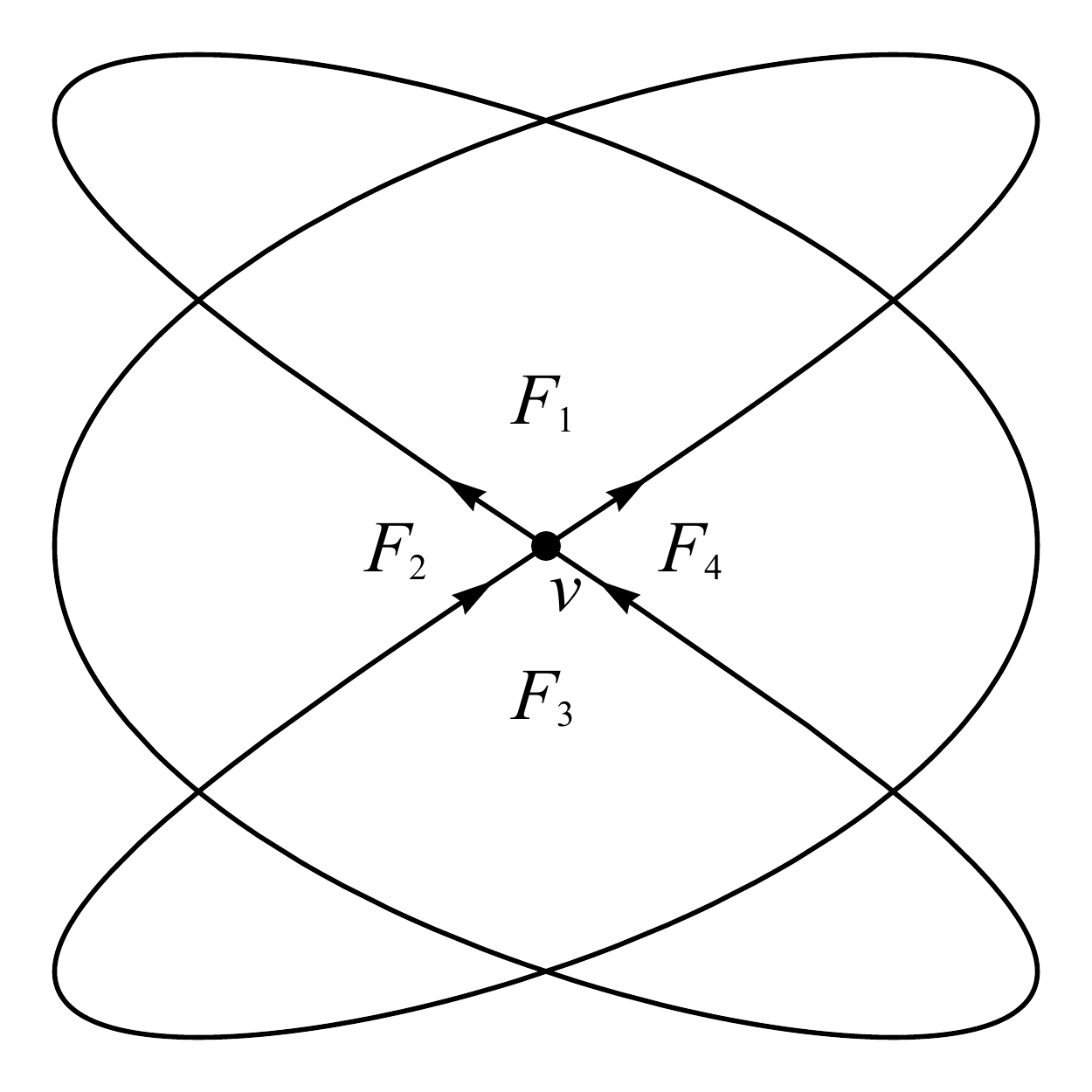}\caption{A typical loop $L$ for the Makeenko--Migdal equation}%
\label{mmplot.fig}%
\end{figure}

We also let $L_{1}$ denote the
loop from the beginning to the first return to $v$ and let $L_{2}$ denote the
loop from the first return to the end. (See Figure \ref{l1l2.fig}.) Finally, we let $t_1,\dots ,t_4$ 
denote the areas of the four faces adjacent to $v$. The
Makeenko--Migdal equation, in the plane case, then gives a formula for the
alternating sum of the derivatives of the Wilson loop functional with respect
to these areas.

\begin{theorem}
[Makeenko--Migdal equation for $U(N)$]\label{MMun.thm}Let $L$ be a
closed\break curve with simple crossings and let $v$ be a crossing.
Parametrize $L$ over the time interval $[0,1]$ with $L(0)=L(1)=v$, and let
$s_{0}$ be the unique time with $0<s_{0}<1$ such that $L(s_{0})=v.$ Let
$L_{1}$ be the restriction of $L$ to $[0,s_{0}]$ and let $L_{2}$ be the
restriction of $L$ to $[s_{0},1].$ Then%
\begin{equation}
\left(  \frac{\partial}{\partial t_{1}}-\frac{\partial}{\partial t_{2}}%
+\frac{\partial}{\partial t_{3}}-\frac{\partial}{\partial t_{4}}\right)
\mathbb{E}\{\mathrm{tr}(\mathrm{hol}(L))\}=\mathbb{E}\{\mathrm{tr}%
(\mathrm{hol}(L_{1}))\mathrm{tr}(\mathrm{hol}(L_{2}))\}. \label{MMun}%
\end{equation}

\end{theorem}

We follow the convention that if any of the adjacent faces is the unbounded
face, the corresponding derivative on the left-hand side of (\ref{MMun}) is
omitted. Note also that the faces $F_{1},$ $F_{2},$ $F_{3},$ and $F_{4}$ are
not necessarily distinct, so that the same derivative may occur more than once
on the left-hand side of (\ref{MMun}). We will actually prove (following
L\'{e}vy) an abstract Makeenko--Migdal equation that allows one to compute
alternating sums of derivatives of more general functions; see Section
\ref{MMexamples.sec} for additional examples.

The original argument of Makeenko and Migdal for the equation that bears their
names was based on heuristic calculations with a path integral and is far from
rigorous. (See Section 0.6 of \cite{Levy}.) Rigorous proofs have been given by
L\'{e}vy \cite{Levy} and Dahlqvist \cite{Dahl}. The goal of the current paper
is to provide three short proofs of the result, each of which is substantially
simpler than the proofs in \cite{Levy} and \cite{Dahl}. Our proofs are
\textquotedblleft local\textquotedblright\ in the sense that the key
calculations involve only the edges and faces adjacent to the crossing $v.$
This local nature of the proofs allows two of them, the proofs based on the
edge variables, to be extended to the case of Yang--Mills theory over
arbitrary compact surfaces; cf.\ \cite{DHKsurf}. In particular, our proofs, in
contrast to those of L\'{e}vy and Dahlqvist, make no reference to the
unbounded face. Since the Makeenko--Migdal equation itself is a local
statement, it is natural that it should have a local proof as well; this was
one motivation for the present paper, which provides purely local arguments.

\begin{figure}[ptb]
\centering
\includegraphics[
height=2.1949in,
width=2.1949in
]{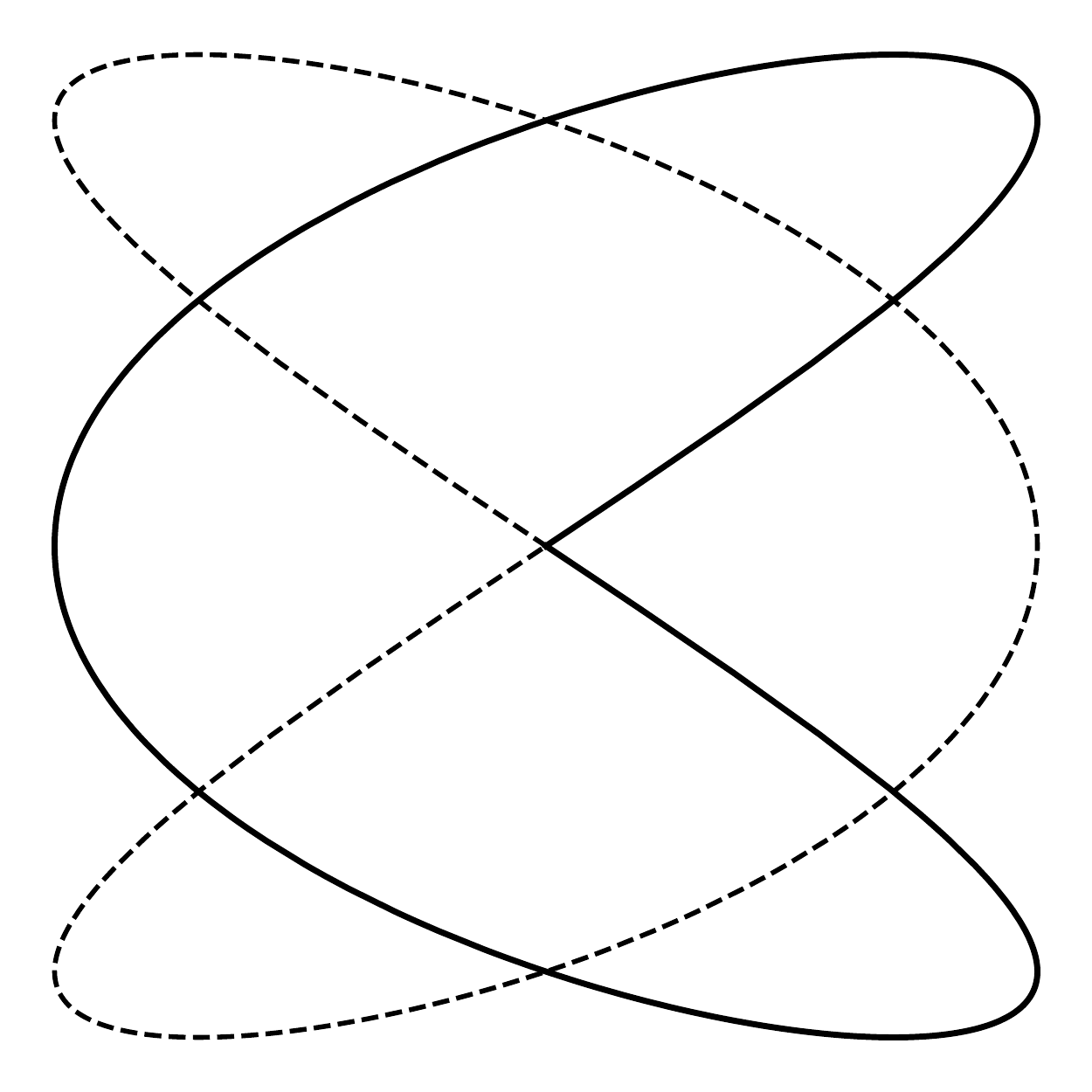}\caption{The loops $L_{1}$ (black) and $L_{2}$ (dashed)}%
\label{l1l2.fig}%
\end{figure}

The significance of (\ref{MMun}) is that the two loops $L_{1}$ and $L_{2}$ on
the right-hand side are simpler than the loop $L$. On the other hand, if one
is attempting to compute Wilson loop expectations recursively, the right-hand
side of (\ref{MMun}) cannot be considered as a \textquotedblleft
known\textquotedblright\ quantity, because it involves the \textit{expectation
of the product} of traces, rather than the \textit{product of the
expectations}. Thus, Theorem \ref{MMun.thm} is not especially useful in
computing Wilson loop expectations for a fixed rank $N$.

Nevertheless, it has been suggested at least since the work of 't~Hooft
\cite{tHooft} that quantum gauge theories with structure group $U(N)$ simplify
in the $N\rightarrow\infty$ limit. Specifically, it has been suggested that in
this limit, the Euclidean Yang--Mills path-integral concentrates onto a
\textit{single} connection (modulo gauge transformations), known as the
\textit{master field}. In the plane case, the structure of the master field
has been described by I. M. Singer \cite{Sing} and R. Gopakumar and D. Gross
\cite{GG,Gop}; see also A. N. Sengupta's paper \cite{Sengupta1}. Recently,
rigorous analyses of the master field on the plane have been given by M.
Anshelevich and A. N. Sengupta \cite{AS} and T. L\'{e}vy \cite{Levy}.
L\'{e}vy, in particular, shows in detail that the Wilson loop functionals
become deterministic in the large-$N$ limit.

In the large-$N$ limit, then, all variances and covariances go to zero,
meaning that there is no difference between an expectation of a product and a
product of the expectations. For the master field on the plane, the
Makeenko--Migdal equation takes the form%
\begin{equation}
\left(  \frac{\partial}{\partial t_{1}}-\frac{\partial}{\partial t_{2}}%
+\frac{\partial}{\partial t_{3}}-\frac{\partial}{\partial t_{4}}\right)
\tau(\mathrm{hol}(L))=\tau(\mathrm{hol}(L_{1}))\tau(\mathrm{hol}(L_{2})),
\label{MMmaster}%
\end{equation}
where $\tau(\cdot)$ is the limiting value of $\mathbb{E}\{\mathrm{tr}%
(\cdot)\}.$ L\'{e}vy shows (Section 9.4 of \cite{Levy}) that by using the
Makeenko--Migdal equation at each crossing of the loop (along with a simpler
relation that we describe in Theorem \ref{simpleDeriv.thm}), one can recover
the derivative of a Wilson loop functional with respect to the area of any one
face. This result leads to an effective procedure for (recursively) computing
the Wilson loop functionals for the master field.

S. Chatterjee has given a rigorous version of the Makeenko--Migdal equation
for lattice gauge theories in any dimension (Theorem 3.6 of \cite{Chatt}).
This equation takes a somewhat different form from the two-dimensional
continuum result in Theorem \ref{MMun.thm}.

\section{Two proofs using edge variables}

\subsection{The set-up}

\label{setup.sec}

The appearance of the heat kernel on $K$ in two-dimensional Yang--Mills theory
can be traced back at least to the work of Migdal. Equation (27) in
\cite{Migdal}, for example, can be understood as the expansion of the heat
kernel on $K$ in terms of characters, although Migdal does not make this
connection explicit. (See also Theorem 2 in \cite{DM} in the $SU(2)$ case.)
Eventually, the role of the heat kernel began to be explicitly identified,
leading to the notion of the \textquotedblleft heat kernel
action,\textquotedblright\ as in work of Menotti and Onofri \cite{MO}. Here,
the heat kernel on $K$ can be used as an alternative to the Wilson action as a
lattice approximation to the continuum Yang--Mills action. In the
two-dimensional case, however, the heat kernel action is invariant under
refinement of the lattice. (That is to say, the heat kernel action in two
dimensions is a \textquotedblleft fixed point for the renormalization group
flow.\textquotedblright) This invariance property of the heat kernel action
suggests that the heat kernel action on a fixed lattice actually gives the
exact continuum result for all Wilson loop variables traced out in that lattice.

A different approach to such results was developed simultaneously by L. Gross,
C. King, and A. Sengupta in \cite{GKS} and B. Driver in \cite{Dr}. These
authors use a gauge fixing to represent the continuum Yang--Mills measure on
the plane as a Gaussian measure, allowing for rigorous computations of Wilson
loop functionals in a continuum theory. Both papers confirm the role of the
heat kernel in the Wilson loop expectations. Driver then gave a formula
\cite[Theorem 6.4]{Dr} for the Wilson loop functional for a self-intersecting
loop traced out on a graph in the plane, under mild restrictions on the nature
of the edges involved. In the work of L\'{e}vy \cite{Levy1}, the author takes
Driver's formula as the definition of the Yang--Mills measure on a graph and
then uses the consistency of this measure under subdivision to construct a
continuous theory.

Driver's formula involves the heat kernel $\rho_{t}$ on the structure group
$K$ with respect to a fixed bi-invariant metric. That is to say, $\rho_{t}$
satisfies the heat equation%
\begin{equation}
\frac{\partial\rho_{t}}{\partial t}=\frac{1}{2}\mathrm{\Delta}\rho_{t},
\label{e.heat.eq}%
\end{equation}
where $\mathrm{\Delta}$ is the Laplacian associated to the given metric, and
for any continuous function $f$ on $K,$ we have%
\[
\lim_{t\rightarrow0}\int_{K}f(x)\rho_{t}(x)\,dx=f(\mathrm{id}),
\]
where $\mathrm{id}$ is the identity element of $K$ and where $dx$ is the
normalized Haar measure on $K.$ It will be important in our computations to
note that the heat kernel with respect to a bi-invariant metric on a compact
Lie group is conjugation invariant:
\begin{equation}
\rho_{t}(uxu^{-1})=\rho_{t}(x),\qquad\forall\;x,u\in K. \label{e.conj.inv}%
\end{equation}
This identity holds because the conjugation action of $K$ on itself is isometric
and fixes the origin.

We now consider the appropriate notion of a graph in the plane. By an edge we
will mean a continuous map $\gamma:[0,1]\rightarrow\mathbb{R}^{2},$ assumed to
be injective except possibly that $\gamma(0)=\gamma(1)$. We identify two edges
if they differ by an orientation-preserving reparametrization. Two edges that
differ by an orientation-reversing reparametrization are said to be inverses
of each other. A graph is then a finite collection of edges (and their
inverses) that meet only at their endpoints. Given a graph $\mathbb{G}$, we
choose arbitrarily one element out of each pair consisting of an edge and its
inverse. We then refer to the chosen edges as the positively oriented edges.

If $n$ denotes the number of positively oriented edges in $\mathbb{G}$,
Driver's result then says that the expectation value of any gauge-invariant
function of the parallel transport along the edges of $\mathbb{G}$ may be
computed as integration against a measure $\mu$ on $K^{n}.$ To compute $\mu,$
we associate to each positively oriented edge $e$ in $\mathbb{G}$ an
\emph{edge variable} $x$ taking values in $K$, and correspondingly associate $x^{-1}$ to the
inverse of $e$. We interpret the edge variable as the parallel transport of a
connection along the edge. Then $\mu$ is given by
\begin{equation}
d\mu=\left(  \prod\rho_{\left\vert F_{i}\right\vert }(h_{i})\right)
~d\mathbf{x}, \label{DriversForm}%
\end{equation}
where the product is over all the bounded faces $F_{i}$ of the graph, that is,
over all the bounded components of the complement of the graph in the plane.
Here $d\mathbf{x}$ denotes the product of normalized Haar measures in all the
edge variables, $\left\vert F_{i}\right\vert $ denotes the area of $F_{i},$
and $h_{i}$ denotes the \textquotedblleft holonomy\textquotedblright\ around
$F_{i},$ that is, the product of edge variables and their inverses going
around the boundary of $F_{i}$; in \cite{MM}, these discrete holonomies were
referred to as plaquettes. Note: since the Haar measure on $K$ is symmetric
(i.e.,\ invariant under $x\mapsto x^{-1} $), the measure $\mu$ is independent
of the choice of which edges in $\mathbb{G}$ are positively oriented.

\begin{remark}
Following the usual conventions in the field, we take the parallel transport
operation ``par" to be \emph{order reversing}. That is, if $\gamma_{1}%
\gamma_{2}$ means ``traverse $\gamma_{1}$ and then traverse $\gamma_{2}$,"
then
\[
\mathrm{par}(\gamma_{1}\gamma_{2}) = \mathrm{par}(\gamma_{2})\mathrm{par}%
(\gamma_{1}).
\]
(The reason for this convention is presumably that the convention for
concatenation of paths is contrary to the usual convention for function
composition, where $f_{1}\circ f_{2}$ means first do $f_{2}$ and then do
$f_{1}$.) The reader should keep in mind this order reversal in the
computations throughout the paper.
\end{remark}

It is harmless to assume that the boundary of each face $F_{i}$ of
$\mathbb{G}$ is connected, although as shown in \cite[pp. 591--592]{Dr} this
assumption is not actually necessary. If the boundary of $F_{i}$ is connected,
it is easy to see that the value of $\rho_{\left\vert F_{i}\right\vert }%
(h_{i})$ does not depend on where one starts on the boundary of $F_{i}$ or on
the direction one proceeds (since the heat kernel has the symmetry $\rho_t(x)=\rho_t(x^{-1})$
for all $x\in K$ and $t>0$).

\begin{definition}
\label{discreteGaugeInv.def} Let $\mathrm{V}(\mathbb{G})$ denote the set of
vertices of $\mathbb{G}$. A \textbf{discrete gauge transformation} is a
map $g:\mathrm{V}(\mathbb{G})\rightarrow K.$ For each discrete gauge
transformation $g,$ we define a transformation $\mathrm{\Psi} _{g}$ of the
edge variables of $\mathbb{G}$ as follows. If $e$ is an edge and $a_{e}$ is
the associated edge variable, we set
\[
\mathrm{\Psi} _{g}(a_{e})=g(v_{2})^{-1}a_{e}g(v_{1}),
\]
where $v_{1}$ and $v_{2}$ are the initial and final vertices of $e,$
respectively. If $f$ is a function of the edge variables, we say that $f$ is
\textbf{gauge invariant} if $f\circ\mathrm{\Psi} _{g}=f$ for every discrete
gauge transformation $g.$
\end{definition}

Note that if an edge $f$ is the inverse of another edge $e$---so that
$a_{f}=a_{e}^{-1}$---then $\mathrm{\Psi}_{g}(a_{f})=(\mathrm{\Psi}_{g}%
(a_{e}))^{-1}.$ Thus, the gauge transformation of the edge variables does not
depend on the choice of orientation of the edges.

If $f$ is a function of the edge variables of $\mathbb{G},$ we can associate a
function $\hat{f}$ on the space of connections for a trivial principal
$K$-bundle over the plane, by defining $\hat{f}%
(A)=f(\mathbf{a}(A)),$ where $\mathbf{a}(A)=\left\{  \mathrm{par}^{A}%
(e):e\in\mathbb{G}\right\}  ,$ and $\mathrm{par}^{A}(e)\in K$ is the parallel
transport of the connection $A$ along the edge $e\in\mathbb{G}.$ If
$g:\mathbb{R}^{2}\rightarrow K$ is a $C^{1}$ \textquotedblleft gauge
transformation\textquotedblright\ and $A^{g}=\mathrm{Ad}_{g}^{-1}A+g^{-1}dg$ is the
connection $A$ gauge transformed by $g,$ then
\[
\mathrm{par}^{A^{g}}(e)=g(v_{2})^{-1}\mathrm{par}^{A}(e)g(v_{1}),
\]
where $v_{1}$ and $v_{2}$ are the initial and final vertices of $e,$
respectively. By comparing this formula to Definition
\ref{discreteGaugeInv.def}, we see that if $f$ is invariant under (discrete)
gauge transformations for $\mathbb{G},$ then $\hat{f}$ is invariant under
(continuous) gauge transformations for $\mathbb{R}^{2}$ and thus constitutes a
valid observable for Yang--Mills theory over $\mathbb{R}^{2}.$ If the edges of
$\mathbb{G}$ satisfy a mild regularity property, Theorem 6.4 of \cite{Dr}
states that
\[
\mathbb{E}\{\hat{f}\}=\int_{K^{n}}f~d\mu,
\]
where the expectation value is with respect to a rigorously defined
Yang--Mills measure on the space of connections modulo gauge transformations.

The measure $\mu$ appears also in the work \cite{Levy} of L\'{e}vy. The
approach there is, however, different, in that L\'{e}vy takes Driver's formula
(\ref{DriversForm}) as the \textit{definition} of the Yang--Mills measure for
a graph in the plane and then uses consistency results to construct a
continuous object. Regardless of the approach used, once the measure $\mu$ has
been defined, it makes sense to integrate \textit{any} function $f$ of the
edge variables, whether or not $f$ has any special invariance property.

\subsection{A simple area-derivative formula}

Before coming to the Makeenko--Migdal equation itself, we record a simple
result that can be proven much more easily. This result was stated by Kazakov
in \cite[Equation 10]{KazakovU(N)}; it is also a special case of Corollary 6.5
of \cite{Levy}, but in this case, the proof simplifies dramatically. We
include a proof here for completeness and to give an indication of the
difficulties in computing area derivatives in general. In \cite{Levy},
L\'{e}vy shows that the master field (i.e.\ the large-$N$ limit of Yang--Mills
for $U(N)$) is completely characterized by the limiting Makeenko--Migdal
equation (\ref{MMmaster}) and the large-$N$ limit of (\ref{simple2}).

\begin{theorem}
[Unbounded Face Condition]\label{simpleDeriv.thm} Suppose $f$ is a smooth
function of the edge variables associated to a graph $\mathbb{G}$ and that $F$
is a bounded face of $\mathbb{G}$ that shares a positively oriented edge $e$
with the unbounded face. Let $a$ denote the edge variable associated to the
edge $e$ and let $t$ denote the area of $F.$ Then we have%
\begin{equation}
\frac{d}{dt}\int f~d\mu=\frac{1}{2}\int\mathrm{\Delta}^{a}f~d\mu,
\label{simple}%
\end{equation}
where $\mathrm{\Delta}^{a}$ denotes the Laplacian with respect to $a$ with the
other edge variables held constant.

In particular, suppose that $K=U(N),$ that $L$ is a loop traced out on
$\mathbb{G}$ in which the edge $e$ (which borders the unbounded face) is
traversed exactly once, and that $f=\mathrm{tr}(\mathrm{hol}(L))$. Then
\eqref{simple} reduces to%
\begin{equation}
\frac{d}{dt}\mathbb{E}\{\mathrm{tr}(\mathrm{hol}(L))\}=-\frac{1}{2}%
\mathbb{E}\{\mathrm{tr}(\mathrm{hol}(L))\}. \label{simple2}%
\end{equation}
Finally, if $K=U(N)$ and $L$ is a simple closed curve enclosing area $t$, we
have
\[
\mathbb{E}\{\mathrm{tr}(\mathrm{hol}(L))\}=e^{-t/2}.
\]

\end{theorem}

The key idea in the proof of (\ref{simple}) is that because the edge $e$ lies
on the boundary of only one bounded face, the edge variable $a$ occurs in only
one of the heat kernels in Driver's formula. By contrast, a generic edge
variable lies in two different heat kernels, which is a substantial
complicating factor in the proof of the Makeenko--Migdal formula.

\begin{proof}
We may choose the orientation of the boundary of $F$ so that it contains the
edge $e$ (as opposed to $e^{-1}$) exactly once. (For example, referring to
Figure \ref{loopvar1.fig} below, we may take $F=F_{5}$ and $e=e_{7}$.) It is
harmless to assume that $e$ is the first edge traversed, in which case, since
parallel transport is order-reversing, the holonomy $h$ around $\partial F$
will have the form%
\[
h=\alpha a,
\]
where $\alpha$ is a word in edge variables other than $a.$ We then note that
$(\mathrm{\Delta}\rho_{t})(h)$ may be computed as%
\[
(\mathrm{\Delta}\rho_{t})(h)=\mathrm{\Delta}^{a}(\rho_{t}(\alpha a)).
\]

Thus, using Driver's formula \eqref{DriversForm} and differentiating under the
integral, we obtain%
\[
\frac{d}{dt}\int f~d\mu=\frac{1}{2}\int f~[\mathrm{\Delta}^{a}\rho_{t}(\alpha
a)]\prod_{F_{i}\neq F}\rho_{\left\vert F_{i}\right\vert }(h_{i})\,d\mathbf{x}%
.
\]
Now, since $e$ lies between $F$ and the unbounded face, the edge variable $a$
does not occur in any other heat kernel besides $\rho_{t}(\alpha a).$ Thus, if
we integrate by parts, the Laplacian does not hit any other heat kernel, but
hits only $f,$ giving (\ref{simple}).

Meanwhile, suppose $K=U(N)$ and $f$ is the normalized trace of the holonomy of
$L,$ where $L$ traverses the edge $e$ exactly once. If $L$ traverses $e$ in
the positive direction, then $f$ will have the form%
\[
f=\mathrm{tr}(\beta a\gamma),
\]
where $\beta$ and $\gamma$ are words in edge variables distinct from $a.$ Then%
\[
\mathrm{\Delta}^{a}f=\sum_{X}\mathrm{tr}(\beta aX^{2}\gamma),
\]
where $X$ ranges over an orthonormal basis for the Lie algebra $\mathfrak{k}%
=\mathfrak{u}(N).$ But a simple argument (e.g., Proposition 3.1 in \cite{DHK})
shows that if the inner product on $\mathfrak{u}(N)$ is normalized as in
(\ref{HSN}) we have%
\[
\sum_{X}X^{2}=-I,
\]
in which case (\ref{simple}) reduces to (\ref{simple2}). If $L$ traverses $e$
negatively, the argument is almost identical. Finally, if $L$ has only one
bounded face with area $t$, Driver's formula (\ref{DriversForm}) tells us that
at $t=0$, the holonomy concentrates at the identity, so that the normalized
trace of the holonomy is 1. \hfill$\square$
\end{proof}

\subsection{An abstract Makeenko--Migdal equation}

Suppose now that $\mathbb{G}$ is a graph in the plane and $v$ is a vertex of
$\mathbb{G}$ having four edges, where we count an edge $e$ twice if both the
initial and final vertices of $e$ are equal to $v$. We label the four
\textit{outgoing} edges at $v$ in cyclic order as $e_{1},$ $e_{2},$ $e_{3},$
and $e_{4}$. Any one of the four outgoing edges may be labeled as $e_{1}$ and
the cyclic ordering may be either clockwise or counter-clockwise. We then
label the four (not necessarily distinct) faces of $\mathbb{G}$ adjacent to
$v$ in cyclic order as $F_{1},$ $F_{2},$ $F_{3},$ and $F_{4}$, where $F_{1}$
is the face whose boundary contains $e_{1}$ and $e_{2}$, $F_{2}$ is the face
whose boundary contains $e_{2}$ and $e_{3}$, etc. Figure \ref{edgesandfaces.fig}
shows one such labeling. 

We assume for now that these edges are distinct; this assumption is removed in
Section \ref{generic.sec}. (More precisely, we assume not only that the
$e_{i}$'s are distinct as oriented edges, but also that $e_{i}\neq e_{j}^{-1}$
for $i\neq j$.) We also let $a_{i}$ denote the edge variable, with values in
$K,$ associated to $e_{i}.$ We write the collection $\mathbf{x}$ of all edge
variables in our graph as
\[
\mathbf{x}=(a_{1},a_{2},a_{3},a_{4},\mathbf{b}),
\]
where $\mathbf{b}$ is the tuple of all edge variables other than $a_{1},$
$a_{2},$ $a_{3},$ and $a_{4}.$ In \cite{Levy}, L\'{e}vy isolates a version of
the Makeenko--Migdal equation that is valid for an arbitrary compact structure
group $K,$ and in which the function does not have to be the trace of a holonomy.

\begin{figure}[ptb]
\centering
\includegraphics[
height=2.1845in,
width=2.1932in
]{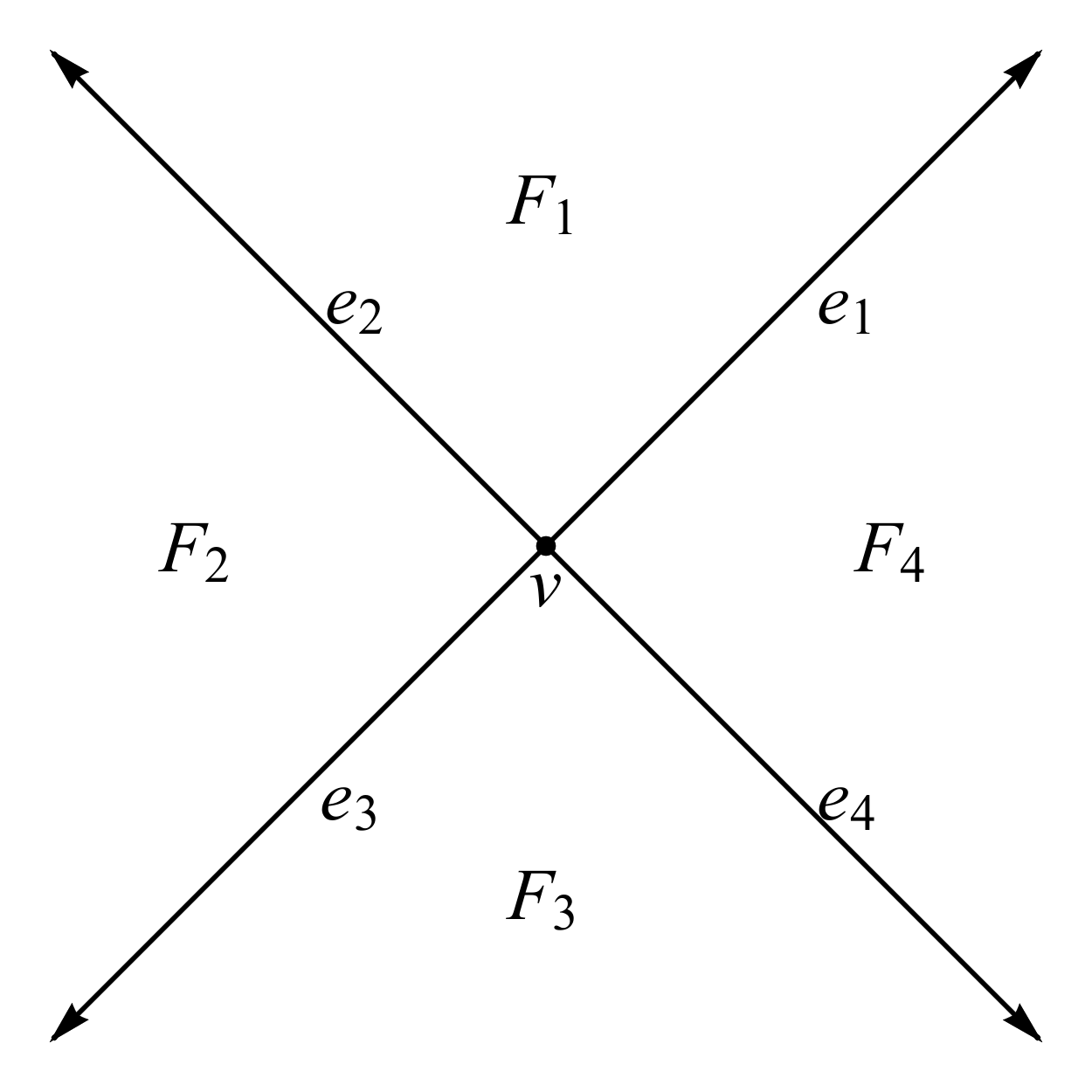}\caption{A possible labeling of the faces and edges adjacent
to $v$}%
\label{edgesandfaces.fig}%
\end{figure}

\begin{definition}
\label{extended.def}If the edges $e_{1},\ldots,e_{4}$ are distinct, a function
$f(a_{1},a_{2},a_{3},a_{4},\mathbf{b})$ of the edge variables has
\emph{extended gauge invariance at }$v$ if, for all $x\in K$,
\begin{equation}
f(a_{1},a_{2},a_{3},a_{4},\mathbf{b})=f(a_{1}x,a_{2},a_{3}x,a_{4}%
,\mathbf{b})=f(a_{1},a_{2}x,a_{3},a_{4}x,\mathbf{b}). \label{extendDef}%
\end{equation}

\end{definition}

By contrast, if the edges are distinct, $f$ has ordinary gauge invariance at
$v$ (i.e., invariance under a gauge transformation supported at the vertex $v
$) if
\begin{equation}
f(a_{1},a_{2},a_{3},a_{4},\mathbf{b})=f(a_{1}x,a_{2}x,a_{3}x,a_{4}%
x,\mathbf{b}) \label{ordinaryInv}%
\end{equation}
for all $x\in K.$ (Apply Definition \ref{discreteGaugeInv.def} with $g(v)=x$
and $g(v^{\prime})=\mathrm{id}$ for all $v\neq v^{\prime}$.) Clearly, extended
gauge invariance at $v$ implies ordinary gauge invariance at $v,$ but not vice versa.

Suppose, for example, that $f$ is the trace of the holonomy around a loop $L$
traced out on $\mathbb{G}$ starting from $v$. We assume $L$ has a simple
crossing at $v$. In that case, as explained in the introduction, it is possible to choose the cyclic ordering of the edges so that $L$ initially leaves $v$ along $e_{1},$ returns along $e_{4}^{-1},$
leaves again along $e_{2},$ and finally returns again along $e_{3}^{-1}$. Thus, $L$ must have the form%
\[
L=e_{1}Ae_{4}^{-1}e_{2}Be_{3}^{-1},
\]
where $A$ and $B$ are sequences of edges not belonging to $\{e_{1},e_{2}%
,e_{3},e_{4}\}.$ Since parallel transport is order-reversing, the trace of the
holonomy around $L$ is then represented by a function of the form
\begin{equation}
f(a_{1},a_{2},a_{3},a_{4},\mathbf{b}):=\mathrm{tr}(a_{3}^{-1}\beta a_{2}%
a_{4}^{-1}\alpha a_{1}), \label{holonomyf}%
\end{equation}
where $\alpha$ and $\beta$ are words in the $\mathbf{b}$ variables. This
function is easily seen to have extended gauge invariance at $v.$

\begin{definition}
\label{gradient.def}If $f$ is a smooth function on $K,$ the
\emph{left-invariant gradient} of $f,$ denoted $\nabla f,$ is the function
with values in the Lie algebra $\mathfrak{k}$ of $K$ given by%
\[
(\nabla f)(x)=\sum_{X}\left(  \left.  \frac{d}{ds}f(xe^{sX})\right\vert
_{s=0}\right)  ~X,
\]
where the sum is over any orthonormal basis of $\mathfrak{k}.$ More generally,
if $f$ is a smooth function of the edge variables and $a$ is one of the edge
variables, we let $\nabla^{a}f$ denote the left-invariant gradient of $f$ with
respect to $a$ with the other variables fixed. Finally, if $a$ and $b$ are two
distinct edge variables, $\nabla^{a}\cdot\nabla^{b}f$ is the scalar-valued
function defined by%
\[
(\nabla^{a}\cdot\nabla^{b}f)(a,b,\mathbf{c})=\sum_{X}\left.  \frac
{\partial^{2}}{\partial s\partial t}f(ae^{sX},be^{tX},\mathbf{c})\right\vert
_{s=t=0}%
\]
where $\mathbf{c}$ is the tuple of edge variables other than $a$ and $b.$
\end{definition}

If $f$ is smooth and has extended gauge invariance at $v$, then by
differentiating (\ref{extendDef}), we obtain%
\begin{equation}
\label{e.inf.ext.gauge.inv}\nabla^{a_{i}}f=-\nabla^{a_{i+2}}f,
\end{equation}
where $i+2$ is computed mod 4. Since, also, $\nabla^{a_{i}}$ commutes with
$\nabla^{a_{j}},$ we have%
\begin{align}
\nabla^{a_{i}}\cdot\nabla^{a_{j}}f  &  =-\nabla^{a_{i}}\cdot\nabla^{a_{j+2}%
}f=-\nabla^{a_{j+2}}\cdot\nabla^{a_{i}}f\nonumber\\
&  =\nabla^{a_{j+2}}\cdot\nabla^{a_{i+2}}f=\nabla^{a_{i+2}}\cdot
\nabla^{a_{j+2}}f, \label{doubleInv}%
\end{align}
even though $\nabla^{a_{j}}f$ does not necessarily have extended gauge invariance.

We are now ready to state (a special case of) L\'{e}vy's abstract form of the
Makeenko--Migdal equation.

\begin{theorem}
[T. L\'{e}vy]\label{MMab.thm}Suppose $\mathbb{G}$ is a graph in the plane and
$v$ is a vertex of $\mathbb{G}$ with four distinct edges emanating from $v.$
Label the four faces of $\mathbb{G}$ adjacent to $v$ in cyclic order as
$F_{1},\ldots,F_{4}$ and label the outgoing edges in cyclic order as
$e_{1},\ldots,e_{4},$ with $e_{1}$ lying between $F_{4}$ and $F_{1}$, $e_{2}$
lying between $F_{1}$ and $F_{2}$, etc. Then if $f$ is a smooth function of
the edge variables of $\mathbb{G}$ having extended gauge invariance at $v$, we
have%
\begin{equation}
\left(  \frac{\partial}{\partial t_{1}}-\frac{\partial}{\partial t_{2}}%
+\frac{\partial}{\partial t_{3}}-\frac{\partial}{\partial t_{4}}\right)  \int
f~d\mu=-\int\nabla^{a_{1}}\cdot\nabla^{a_{2}}f~d\mu, \label{e.MM.abstract}%
\end{equation}
where $t_{i}$ is the area of $F_{i},$ $i=1,\ldots,4.$
\end{theorem}

As usual, we set $\partial/\partial t_{i}$ equal to zero if $F_{i}$ is the
unbounded face. A version of the theorem still holds even if the edges
$e_{1},\ldots,e_{4}$ are not distinct; see Section \ref{generic.sec}. Note
that $f$ is not assumed to have any special invariance property at any vertex
other than $v.$

Theorem \ref{MMab.thm} is a special case of Proposition 6.22 in \cite{Levy}.
Specifically, since the Yang--Mills measure does not depend on the orientation
of the plane, it is harmless to assume that the faces $F_{1},$ $F_{2},$
$F_{3},$ $F_{4}$ in our labeling scheme occur in counterclockwise order, as in
Figure \ref{mmplot.fig}. We may take the set $I$ in Levy's Proposition 6.22 to
be $\{e_{1},e_{3}\},$ as in Figure 25 in \cite{Levy}. Then the left-hand side
of Proposition 6.22 is actually the negative of the usual alternating sum of
area-derivatives. On the right-hand side of Proposition 6.22, meanwhile, there
is only one term in the sum, namely $\int\mathrm{\Delta}^{e_{1};e_{2}}f~d\mu$,
which corresponds to $\int\nabla^{a_{1}}\cdot\nabla^{a_{2}}f~d\mu$ in our notation.

Note that since $f$ is assumed to have extended gauge invariance at $v$, we
have, as in (\ref{doubleInv}),%
\begin{equation}
\nabla^{a_{1}}\cdot\nabla^{a_{2}}f=-\nabla^{a_{2}}\cdot\nabla^{a_{3}}%
f=\nabla^{a_{3}}\cdot\nabla^{a_{4}}f=-\nabla^{a_{4}}\cdot\nabla^{a_{1}}f.
\label{gradABforms}%
\end{equation}
If we specialize Theorem \ref{MMab.thm} to the case in which $K=U(N)$ and $f$
is as in (\ref{holonomyf}), we find that%
\begin{align}
&  \left(  \frac{\partial}{\partial t_{1}}-\frac{\partial}{\partial t_{2}%
}+\frac{\partial}{\partial t_{3}}-\frac{\partial}{\partial t_{4}}\right)
\int\mathrm{tr}(a_{3}^{-1}\beta a_{2}a_{4}^{-1}\alpha a_{1})~d\mu\nonumber\\
&  =-\sum_{X}\int\mathrm{tr}(a_{3}^{-1}\beta a_{2}Xa_{4}^{-1}\alpha
a_{1}X)~d\mu, \label{MMholonomy}%
\end{align}
where the sum is over any orthonormal basis $\{X\}$ for $\mathfrak{u}(N)$. But
an elementary argument (e.g. \cite[Proposition 3.1]{DHK}) shows that if we
normalize the inner product on $\mathfrak{u}(N)$ as in (\ref{HSN}), then
\begin{equation}
\sum_{X}XCX=-\mathrm{tr}(C)I \label{magic}%
\end{equation}
for any $N\times N$ matrix $C.$ Thus, (\ref{MMholonomy}) reduces to
\begin{align}
&  \left(  \frac{\partial}{\partial t_{1}}-\frac{\partial}{\partial t_{2}%
}+\frac{\partial}{\partial t_{3}}-\frac{\partial}{\partial t_{4}}\right)
\int\mathrm{tr}(a_{3}^{-1}\beta a_{2}a_{4}^{-1}\alpha a_{1})~d\mu\nonumber\\
&  =\int\mathrm{tr}(a_{4}^{-1}\alpha a_{1})\mathrm{tr}(a_{3}^{-1}\beta
a_{2})~d\mu, \label{splitTrace}%
\end{align}
which is---in light of Driver's formula---just the Makeenko--Migdal equation
for $U(N),$ as in Theorem \ref{MMun.thm}.

The goal of this section is to give two short proofs of Theorem \ref{MMab.thm}%
. In \cite{Levy}, L\'{e}vy develops a method of differentiating any function
with respect to the area $t_{i}$ of some face $F_{i}.$ Specifically, if $f$ is
any smooth function of the edge variables---which need not have any special
invariance property---L\'{e}vy shows that%
\begin{equation}
\frac{\partial}{\partial t_{i}}\int f~d\mu=\int Df~d\mu,
\label{LevyDerivative}%
\end{equation}
where $D$ is a certain differential operator. (See Corollary 6.5 in
\cite{Levy}.) The formula for $D$ involves the choice of a spanning tree in
$\mathbb{G}$ and a sum over a sequence of adjacent faces proceeding from
$F_{i}$ to the unbounded face. Thus, $D$ contains, in general, derivatives
involving edges far from the vertex in question.

L\'{e}vy then specializes his result to the case where $f$ has extended gauge
invariance at $v$ and takes the alternating sum of derivatives around a
vertex. At that point, a substantial cancellation occurs: all derivatives of
$f$ drop out, except for derivatives involving edges coming out of the
crossing, and L\'{e}vy then obtains the abstract Makeenko--Migdal equation of
Theorem \ref{MMab.thm}. (See the proof of Proposition 6.22 in \cite{Levy}.)

Our strategy for a simplified proof of Theorem \ref{MMab.thm} is to think that
if the cancellation described in the previous paragraph actually occurs, it
should be possible to see the cancellation \textquotedblleft
locally,\textquotedblright\ that is, in such a way that derivatives involving
far away edges never occur in the first place. Of course, L\'{e}vy's formula
(\ref{LevyDerivative}) is useful for various computations, notably the
computation of finite-$N$ Wilson loop variables in Section 6.9 of \cite{Levy}.
Nevertheless, we do not use this result in our proofs of the Makeenko--Migdal
equation \eqref{MMun}. The local nature of our argument allows us to prove a
new rigorous result, namely that the Makeenko--Migdal equation holds also for
Yang--Mills theory on an arbitrary compact surface, as shown in our paper
\cite{DHKsurf}. In particular, our argument in the plane case does not make
any reference to the unbounded face.

It is interesting to note that although the function $f$ in Theorem
\ref{MMab.thm} is assumed to have extended gauge invariance at $v,$ the
function $\nabla^{a_{1}}\cdot\nabla^{a_{2}}f$ occurring on the right-hand side
of (\ref{e.MM.abstract}), does not necessarily have this invariance. In
(\ref{splitTrace}), for example, the product of traces on the right-hand side
is not invariant under the transformation sending $a_{1}$ to $a_{1}x$ and
$a_{3} $ to $a_{3}x$ (while leaving the other edge variables unchanged). On
the other hand, as discussed in Section \ref{setup.sec}, the most natural
application of Theorem \ref{MMab.thm} is to a function $f$ that is gauge
invariant (in addition to having extended gauge invariance at $v$). In that
case, it is natural to expect that the function $\nabla^{a_{1}}\cdot
\nabla^{a_{2}}f$ will also be gauge invariant, so that the right-hand side of
(\ref{e.MM.abstract}) can be interpreted as the expectation value of a
functional on the space of connections modulo gauge transformations with
respect to the Yang--Mills measure. This expectation is fulfilled in the
following result.

\begin{proposition}
\label{gaugeInvPreserved.prop}Suppose $f$ is a function of the edge variables
of $\mathbb{G}$ that is gauge invariant in the sense of Definition
\ref{discreteGaugeInv.def}. Let $v$ be a vertex of $\mathbb{G}$, let $e$ and
$f$ be two distinct outgoing edges of $\mathbb{G}$ at $v$ such that $e\neq
f^{-1},$ and let $a$ and $b$ be the associated edge variables. Then the
function $\nabla^{a}\cdot\nabla^{b}f$ is also gauge invariant.
\end{proposition}

A version of this proposition holds even if $e=f^{-1}$; see Remark
\ref{nongenericGauge.rem} below.

\begin{proof} For each~$X\in\mathfrak{k}$ and $s,t\in\mathbb{R},$ define a transformation
$\mathrm{\Phi}_{s,t}^{X}:K^{n}\rightarrow K^{n}$ of the edge variables by
replacing $a$ by $ae^{sX}$ and $b$ by $be^{tX},$ while leaving all other edge
variables unchanged. Then
\begin{equation}
\nabla^{a}\cdot\nabla^{b}f=\sum_{X}\left.  \frac{\partial^{2}}{\partial
s~\partial t}f\circ\mathrm{\Phi}_{s,t}^{X}\right\vert _{s=t=0}
\label{gradabComputation}%
\end{equation}
where the sum runs over $X$ in an orthonormal basis of $\mathfrak{k}.$ For
each discrete gauge transformation $g,$ let $\mathrm{\Psi}_{g}:K^{n}%
\rightarrow K^{n}$ be the associated transformation of the edge variables, as
in Definition \ref{discreteGaugeInv.def}. We claim that for all $s,$ $t,$ $X,$
and $g,$ the following identity holds:
\begin{equation}
\mathrm{\Phi}_{s,t}^{X}\circ\mathrm{\Psi}_{g}=\mathrm{\Psi}_{g}\circ
\mathrm{\Phi}_{s,t}^{\mathrm{Ad}_{g(v)}(X)}. \label{gaugeIdentity}%
\end{equation}

To verify (\ref{gaugeIdentity}), note that
\[
\mathrm{\Psi}_{g}(a,b,\mathbf{c})=(g(v_{2})ag(v),g(v_{2}^{\prime
})bg(v),\mathbf{c}_{g}),
\]
where $v_{2}$ and $v_{2}^{\prime}$ are the final vertices of $e$ and $f,$
respectively, $\mathbf{c}$ denotes the collection of edge variables distinct
from $a$ and $b$, and $\mathbf{c}_{g}$ denotes the $\mathbf{c}$-variables
transformed by the discrete gauge transform, $g.$ Thus,%
\begin{align*}
\mathrm{\Phi}_{s,t}^{X}\circ\mathrm{\Psi}_{g}(a,b,\mathbf{c})  &
=(g(v_{2})ag(v)e^{sX},g(v_{2}^{\prime})bg(v)e^{tX},\mathbf{c}_{g})\\
&  =(g(v_{2})ag(v)e^{sX}g(v)^{-1}g(v),g(v_{2}^{\prime})bg(v)e^{tX}%
g(v)^{-1}g(v),\mathbf{c}_{g})\\
&  =(g(v_{2})ae^{s\mathrm{Ad}_{g(v)}(X)}g(v),g(v_{2}^{\prime})be^{t\mathrm{Ad}%
_{g(v)}(X)}g(v),\mathbf{c}_{g})\\
&  =\mathrm{\Psi}_{g}\circ\mathrm{\Phi}_{s,t}^{\mathrm{Ad}_{g(v)}%
(X)}(a,b,\mathbf{c}).
\end{align*}

With the identity (\ref{gaugeIdentity}) in hand, we compute that%
\begin{align*}
(\nabla^{a}\cdot\nabla^{b}f)\circ\mathrm{\Psi}_{g}  &  =\sum_{X}\left.
\frac{\partial^{2}}{\partial s~\partial t}\left[  f\circ\mathrm{\Phi}%
_{s,t}^{X}\circ\mathrm{\Psi}_{g}\right]  \right\vert _{s=t=0}\\
&  =\sum_{X}\left.  \frac{\partial^{2}}{\partial s~\partial t}\left[
f\circ\mathrm{\Psi}_{g}\circ\mathrm{\Phi}_{s,t}^{\mathrm{Ad}_{g(v)}%
(X)}\right]  \right\vert _{s=t=0}\\
&  =\nabla^{a}\cdot\nabla^{b}\left[  f\circ\mathrm{\Psi}_{g}\right]  ,
\end{align*}
where in the last equality we have used that $\mathrm{Ad}_{g(v)}(X)$ runs over an
orthonormal basis of $\mathfrak{k}$ when $X$ does. If $f$ is gauge invariant
(i.e., $f=f\circ\mathrm{\Psi}_{g}),$ then the previously displayed equation
reduces to $(\nabla^{a}\cdot\nabla^{b}f)\circ\mathrm{\Psi}_{g}=\nabla^{a}%
\cdot\nabla^{b}f$ which shows $\nabla^{a}\cdot\nabla^{b}f$ is also gauge
invariant.\hfill$\square$
\end{proof}

\begin{remark}
\label{nongenericGauge.rem}If the edge $e$ is equal to $f^{-1},$ then the edge
variable $b$ associated to $f$ is not an indepedent variable from $a.$ In that
case, as explained in Section \ref{interpret.sec}, the natural way to define
$\nabla^{a}\cdot\nabla^{b}f$ is to use (\ref{gradabComputation}), but where
now $\mathrm{\Phi} _{s,t}^{X}$ replaces $a$ by $e^{-tX}ae^{sX}$ and leaves all
other (independent) edge variables unchanged. With this definition,
Proposition \ref{gaugeInvPreserved.prop} still holds, with a small
modification of the preceding proof.
\end{remark}

\subsection{Two \textquotedblleft local\textquotedblright\ proofs of the
theorem\label{edge.sec}}

We consider at first the \textquotedblleft generic\textquotedblright\ case, in
which the faces $F_{1},$ $F_{2},$ $F_{3}$, and $F_{4}$ are distinct and
bounded, and the edges $e_{1},$ $e_{2},$ $e_{3},$ and $e_{4}$ from $v$ are
distinct. (These assumptions are lifted in Section \ref{generic.sec}.) In that
case, the boundary of $F_{i}$ may be represented by a loop of the form%
\begin{equation}
\partial F_{i}=e_{i}A_{i}e_{i+1}^{-1}, \label{boundaryF}%
\end{equation}
where $A_{i}$ is a sequence of edges not belonging to $\{e_{1},e_{2}%
,e_{3},e_{4}\},$ where the index $i$ is understood to be in $\mathbb{Z}/4.$
Since parallel transport is order reversing, the holonomy $h_{i}$ around
$\partial F_{i}$ is represented by an expression of the form%
\begin{equation}
h_{i}=a_{i+1}^{-1}\alpha_{i}a_{i},\quad i=1,2,3,4, \label{hi}%
\end{equation}
where $\alpha_{i}$ is a word in the $\mathbf{b}$ variables (i.e., the edge
variables not belonging to $\{a_{1},a_{2},a_{3},a_{4}\}$). Furthermore, none
of the variables $a_{1},a_{2},a_{3},a_{4}$ appears in any holonomy other than
ones associated to $F_{1},F_{2},F_{3},F_{4}$. Thus, the Yang--Mills measure
$\mu$ takes the form%
\begin{equation}
d\mu=\rho_{t_{1}}(a_{2}^{-1}\alpha_{1}a_{1})\rho_{t_{2}}(a_{3}^{-1}\alpha
_{2}a_{2})\rho_{t_{3}}(a_{4}^{-1}\alpha_{3}a_{3})\rho_{t_{4}}(a_{1}^{-1}%
\alpha_{4}a_{4})\nu(\mathbf{b})~d\mathbf{x,} \label{muInAVariables}%
\end{equation}
where $d\mathbf{x}$ is the product of the normalized Haar measures in all the
edge variables, and $\nu(\mathbf{b})$ is a product of heat kernels in
$\mathbf{b}$ variables.

Our proofs based on the edge variables are based on the following
\textquotedblleft local\textquotedblright\ version of the abstract
Makeenko--Migdal equation. Since the local structure of the Yang--Mills
measure on an arbitrary compact surface is the same as on the plane, Theorem
\ref{localMM.thm} can be applied also on surfaces. This observation leads to a
proof of the Makeenko--Migdal equation over surfaces, as worked out in
\cite{DHKsurf}.

\begin{theorem}
[Local Abstract Makeenko--Migdal Equation]\label{localMM.thm}Suppose
$f:K^{4}\rightarrow\mathbb{C}$ is a smooth function satisfying the following
\textquotedblleft extended gauge invariance\textquotedblright\ property:%
\[
f(a_{1},a_{2},a_{3},a_{4})=f(a_{1}x,a_{2},a_{3}x,a_{4})=f(a_{1},a_{2}%
x,a_{3},a_{4}x)
\]
for all $\mathbf{a}=(a_{1},a_{2},a_{3},a_{4})$ in $K^{4}$ and all $x$ in $K.$
For each fixed $\mathbf{\alpha}=(\alpha_{1},\alpha_{2},\alpha_{3},\alpha_{4})
$ in $K^{4}$ and $\mathbf{t}=(t_{1},t_{2},t_{3},t_{4})$ in $(\mathbb{R}%
^{+})^{4},$ define a measure $\mu_{\mathbf{\alpha},\mathbf{t}}$ on $K^{4}$ by%
\begin{equation}
\label{e.def.mu_alpha}d\mu_{\mathbf{\alpha},\mathbf{t}}(\mathbf{a}%
)=\rho_{t_{1}}(a_{2}^{-1}\alpha_{1}a_{1})\rho_{t_{2}}(a_{3}^{-1}\alpha
_{2}a_{2})\rho_{t_{3}}(a_{4}^{-1}\alpha_{3}a_{3})\rho_{t_{4}}(a_{1}^{-1}%
\alpha_{4}a_{4})~d\mathbf{a},
\end{equation}
where $d\mathbf{a}$ is the normalized Haar measure on $K^{4}.$ Then for all
$\mathbf{\alpha}\in K^{4},$ we have
\begin{equation}
\left(  \frac{\partial}{\partial t_{1}}-\frac{\partial}{\partial t_{2}}%
+\frac{\partial}{\partial t_{3}}-\frac{\partial}{\partial t_{4}}\right)
\int_{K^{4}}f~d\mu_{\mathbf{\alpha},\mathbf{t}}=-\int_{K^{4}}\nabla^{a_{1}%
}\cdot\nabla^{a_{2}}f~d\mu_{\mathbf{\alpha},\mathbf{t}}. \label{mmLocal}%
\end{equation}

\end{theorem}

Our first proof of Theorem \ref{localMM.thm} proceeds by directly computing
the alternating sum of area-derivatives, and integrating by parts twice. Our
second proof, which is even shorter, proceeds from the right-hand-side of
(\ref{mmLocal}) and relies on the decomposition of the density of $\mu
_{\alpha,\mathbf{t}}$ into the product of $(t_{1},t_{2})$ heat kernels (both
independent of edge variable $a_{4}$) and $(t_{3},t_{4})$ heat kernels (both
independent of edge variable $a_{2}$).

We now observe that Theorem \ref{localMM.thm} easily implies the generic case
of the abstract Makeenko--Migdal equation in Theorem \ref{MMab.thm}.

\begin{proof}
[of Theorem \ref{MMab.thm} (Generic Case)]If $f:K^{n}\rightarrow\mathbb{C}$
has extended gauge invariance at $v$, then $f(a_{1},a_{2},a_{3},a_{4}%
,\mathbf{b})$ has extended gauge invariance as a function of $a_{1}%
,\ldots,a_{4}$ for each $\mathbf{b}.$ In light of (\ref{hi}), we will have%
\[
\int_{K^{n}}f~d\mu=\int_{K^{n-4}}\int_{K^{4}}f(\mathbf{a},\mathbf{b}%
)~d\mu_{\mathbf{\alpha},\mathbf{t}}(\mathbf{a})~\nu(\mathbf{b})~d\mathbf{b},
\]
where $\nu(\mathbf{b})$ is a product of heat kernels in the $\mathbf{b}$
variables. Since the only dependence on $(t_{1},t_{2},t_{3},t_{4})$ in the
integral is in $\mu_{\mathbf{\alpha},\mathbf{t}},$ the time derivatives in
Theorem \ref{MMab.thm} will pass over the outer integral and hit on the
integral over $K^{4}.$ Theorem \ref{MMab.thm} then follows from Theorem
\ref{localMM.thm}. \hfill$\square$
\end{proof}

We also prove Theorem \ref{MMab.thm}, for gauge invariant functions, in
Section \ref{loop.sec} using the loop or lasso variables. This third proof is
also in a sense local.

It remains to prove the local result in Theorem \ref{localMM.thm}.

\subsubsection{First proof of Theorem \ref{localMM.thm}}

Our strategy is to differentiate under the integral sign, use the heat
equation satisfied by the heat kernel, and then integrate by parts. In this
process, we will get \textquotedblleft good terms\textquotedblright\ in which
derivatives hit on the function $f$, and \textquotedblleft bad
terms\textquotedblright\ in which derivatives hit on other heat kernels. In
each of the two stages of integration by parts, we obtain a cancellation of
the bad terms, allowing all of the derivatives to move off of the heat kernels
and onto $f,$ at which point we easily obtain the local Makeenko--Migdal
equation in \eqref{mmLocal}.

To begin, when computing the time derivatives on the left-hand-side of
\eqref{mmLocal}, passing under the integral, we will use the heat equation
\eqref{e.heat.eq} for $\rho_{t_{j}}$; thus, we must deal with terms of the
form $(\mathrm{\Delta}\rho_{t_{j}})(a_{j+1}^{-1}\alpha_{j}a_{j})$. Since the
heat kernel on $K$ is invariant under conjugation, we can compute these by
various combinations of derivatives with respect to $a_{j}$ and derivatives
with respect to $a_{j+1}$. The following lemma yields a convenient way to
express these terms.

\begin{lemma}
\label{l.Brian's.proof.1} For any $j\in \{1,2,3,4\}$ and any smooth conjugation-invariant
function $\rho\colon K\rightarrow\mathbb{C}$,
\begin{equation}
(\mathrm{\Delta}\rho)(a_{j+1}^{-1}\alpha_{j}a_{j})=\frac{1}{4}\left(
\nabla^{a_{j}}-\nabla^{a_{j+1}}\right)  ^{2}[\rho(a_{j+1}^{-1}\alpha_{j}%
a_{j})]. \label{e.l.Brian's.proof.1}%
\end{equation}

\end{lemma}

\begin{proof}
To make the notation precise:
\[
(\nabla^{a_{j}}-\nabla^{a_{j+1}})^{2}\equiv\sum_{X}\left(  \widehat{X}^{a_{j}%
}-\widehat{X}^{a_{j+1}}\right)  ^{2},
\]
where $\hat{X}$ denotes the left-invariant vector field associated to
$X\in\mathfrak{k}$ and where the sum is over any orthonormal basis of
$\mathfrak{k}$, as in Definition \ref{gradient.def}. Now, on the one hand, we
have
\begin{align*}
\widehat{X}^{a_{j}}[\rho(a_{j+1}^{-1}\alpha_{j}a_{j})]  &  =\left.
\frac{\partial}{\partial t}\rho\left(  a_{j+1}^{-1}\alpha_{j}(a_{j}%
e^{tX})\right)  \right\vert _{t=0}\\
&  =\left.  \frac{\partial}{\partial t}\rho\left(  (a_{j+1}^{-1}\alpha
_{j}a_{j})e^{tX}\right)  \right\vert _{t=0}=(\widehat{X}\rho)(a_{j+1}%
^{-1}\alpha_{j}a_{j}).
\end{align*}
On the other hand, by the conjugation invariance of $\rho$, we have
\begin{align*}
\widehat{X}^{a_{j+1}}[\rho(a_{j+1}^{-1}\alpha_{j}a_{j})]  &  =\left.
\frac{\partial}{\partial t}\rho\left(  (a_{j+1}e^{tX})^{-1}\alpha_{j}%
a_{j}\right)  \right\vert _{t=0}\\
&  =\left.  \frac{\partial}{\partial t}\rho\left(  (a_{j+1}^{-1}\alpha
_{j}a_{j})e^{-tX}\right)  \right\vert _{t=0}=-(\widehat{X}\rho)(a_{j+1}%
^{-1}\alpha_{j}a_{j}).
\end{align*}
Thus $(\widehat{X}^{a_{j}}-\widehat{X}^{a_{j+1}})[\rho(a_{j+1}^{-1}\alpha
_{j}a_{j})]=2(\hat{X}\rho)(a_{j+1}^{-1}\alpha_{j}a_{j})$.

Now, although the function $\hat{X}\rho$ may not be conjugation invariant, it
is easily seen to be invariant under conjugation by elements of the form
$e^{tX}$, which is all that is needed in the argument in the previous
paragraph. Thus, applying $\widehat{X}^{a_{j}}-\widehat{X}^{a_{j+1}}$ a second
time gives
\[
\left(  \widehat{X}^{a_{j}}-\widehat{X}^{a_{j+1}}\right)  ^{2}[\rho
(a_{j+1}^{-1}\alpha_{j}a_{j})]=4(\hat{X}^{2}\rho)(a_{j+1}^{-1}\alpha_{j}%
a_{j}).
\]
Summing over $X$ yields the lemma. \hfill$\square$
\end{proof}

We now proceed with the proof of Theorem \ref{localMM.thm}. Denote the density
of $\mu_{\mathbf{\alpha},\mathbf{t}}$ (cf.\ \eqref{e.def.mu_alpha}) by
$\mathcal{R}$:
\begin{equation}
\mathcal{R}(\mathbf{a})=\rho_{t_{1}}(a_{2}^{-1}\alpha_{1}a_{1})\rho_{t_{2}%
}(a_{3}^{-1}\alpha_{2}a_{2})\rho_{t_{3}}(a_{4}^{-1}\alpha_{3}a_{3})\rho
_{t_{4}}(a_{1}^{-1}\alpha_{4}a_{4}). \label{e.Brian.density}%
\end{equation}
For ease of reading, denote $b_{j}=a_{j+1}^{-1}\alpha_{j}a_{j}$, and denote by
$\rho_{j}$ the function $\rho_{j}(\mathbf{a})=\rho_{t_{j}}(a_{j+1}^{-1}%
\alpha_{j}a_{j})=\rho_{t_{j}}(b_{j})$, where the index $j$ is mod $4$ as
usual. Thus $\mathcal{R}=\rho_{1}\rho_{2}\rho_{3}\rho_{4}$. Starting with each
time-derivative term on the left-hand-side of \eqref{mmLocal}, we pass the
derivative under the integral. Using the heat equation \eqref{e.heat.eq}, and
applying Lemma \ref{l.Brian's.proof.1} with $\rho=\rho_{j}$, we have%

\[
\frac{\partial}{\partial t_{j}}\int f\,d\mu_{\mathbf{\alpha},\mathbf{t}}%
=\frac{1}{2}\int f\,(\mathrm{\Delta}\rho_{j})\frac{\mathcal{R}}{\rho_{j}%
}\,d\mathbf{a}=\frac{1}{8}\int f\,\left[  (\nabla^{a_{j}}-\nabla^{a_{j+1}%
})^{2}\rho_{j}\right]  \frac{\mathcal{R}}{\rho_{j}}\,d\mathbf{a}.
\]

We now move one factor of $\nabla^{a_{j}}-\nabla^{a_{j+1}}$ off of $\rho_{j}$,
integrating by parts a first time:
\begin{align}
&  \int f\,\left[  (\nabla^{a_{j}}-\nabla^{a_{j+1}})^{2}\rho_{j}\right]
\frac{\mathcal{R}}{\rho_{j}}\,d\mathbf{a}\nonumber\\
=  &  -\int(\nabla^{a_{j}}-\nabla^{a_{j+1}})\rho_{j}\cdot(\nabla^{a_{j}%
}-\nabla^{a_{j+1}})\left(  f\frac{\mathcal{R}}{\rho_{j}}\right)
\,d\mathbf{a}. \label{e.Brian.int.by.parts.1}%
\end{align}
Following the proof of Lemma \ref{l.Brian's.proof.1}, we can express the first
factor on the right-hand side of (\ref{e.Brian.int.by.parts.1}) as
\begin{equation}
(\nabla^{a_{j}}-\nabla^{a_{j+1}})\rho_{j}=(\nabla^{a_{j}}-\nabla^{a_{j+1}%
})[\rho_{t_{j}}(b_{j})]=2(\nabla\rho_{t_{j}})(b_{j}). \label{firstTerm}%
\end{equation}
Using the product rule in the second factor gives
\[
(\nabla^{a_{j}}-\nabla^{a_{j+1}})\left(  f\frac{\mathcal{R}}{\rho_{j}}\right)
=\frac{\mathcal{R}}{\rho_{j}}(\nabla^{a_{j}}-\nabla^{a_{j+1}})f+f(\nabla
^{a_{j}}-\nabla^{a_{j+1}})\left(  \frac{\mathcal{R}}{\rho_{j}}\right)  .
\]

Now, $\mathcal{R}/\rho_{j}$ consists of the three heat kernel terms other than
$\rho_{j}$. The only one among these that depends on $a_{j}$ is $\rho
_{j-1}=\rho_{t_{j-1}}(a_{j}^{-1}\alpha_{j}a_{j-1})$; the only one that depends
on $a_{j+1}$ is $\rho_{j+1}=\rho_{t_{j+1}}(a_{j+2}^{-1}\alpha_{j+1}a_{j+1})$.
Thus
\begin{align}
(\nabla^{a_{j}}-\nabla^{a_{j+1}})\frac{\mathcal{R}}{\rho_{j}}  &
=(\nabla^{a_{j}}\rho_{j-1})\frac{\mathcal{R}}{\rho_{j-1}\rho_{j}}%
-(\nabla^{a_{j+1}}\rho_{j+1})\frac{\mathcal{R}}{\rho_{j}\rho_{j+1}}\nonumber\\
&  =-\left[  \frac{(\nabla\rho_{t_{j-1}})(b_{j-1})}{\rho_{j-1}\rho_{j}}%
+\frac{(\nabla\rho_{t_{j+1}})(b_{j+1})}{\rho_{j}\rho_{j+1}}\right]
\mathcal{R}. \label{secondTerm}%
\end{align}
Substituting (\ref{firstTerm}) and (\ref{secondTerm}) into
(\ref{e.Brian.int.by.parts.1}) gives three terms:
\begin{align*}
\frac{\partial}{\partial t_{j}}\int f\,d\mu=-  &  A_{j}+B_{j}+B_{j+1}%
\end{align*}
where
\begin{align*}
A_{j}  &  \equiv\frac{1}{8}\int(\nabla^{a_{j}}-\nabla^{a_{j+1}})\rho_{j}%
\cdot\lbrack(\nabla^{a_{j}}-\nabla^{a_{j+1}})f]\frac{1}{\rho_{j}}%
\,d\mu_{\alpha,\mathbf{t}},\quad\text{and}\\
B_{j}  &  \equiv\frac{1}{4}\int\frac{(\nabla\rho_{t_{j}})(b_{j})}{\rho_{j}%
}\cdot\frac{(\nabla\rho_{t_{j-1}})(b_{j-1})}{\rho_{j-1}}f\,d\mu_{\alpha
,\mathbf{t}}.
\end{align*}
Upon taking the alternating sum, the (\textquotedblleft bad\textquotedblright)
$B_{j}+B_{j+1}$ terms cancel in pairs, leaving only the (\textquotedblleft
good\textquotedblright) $A_{j}$ terms:
\begin{align}
&  \sum_{j=1}^{4}(-1)^{j-1}\frac{\partial}{\partial t_{j}}\int f\,d\mu
_{\alpha,\mathbf{t}}\nonumber\\
=  &  -\frac{1}{8}\sum_{j=1}^{4}(-1)^{j-1}\int(\nabla^{a_{j}}-\nabla^{a_{j+1}%
})\rho_{j}\cdot\lbrack(\nabla^{a_{j}}-\nabla^{a_{j+1}})f]\frac{1}{\rho_{j}%
}\,d\mu_{\alpha,\mathbf{t}}. \label{firstStage}%
\end{align}

We now integrate by parts again in each of the integrals in (\ref{firstStage}%
), moving the remaining derivatives off of $\rho_{j}$ and onto the other
factors. Using the product rule, this gives
\begin{align}
&  \int(\nabla^{a_{j}}-\nabla^{a_{j+1}})\rho_{j}\cdot\lbrack(\nabla^{a_{j}%
}-\nabla^{a_{j+1}})\,f]\frac{\mathcal{R}}{\rho_{j}}\,d\mathbf{a}\\
=  &  -\int(\nabla^{a_{j}}-\nabla^{a_{j+1}})^{2}f\,d\mu_{\alpha,\mathbf{t}%
}\label{e.Brian.int.by.parts.1.5}\\
&  - \int(\nabla^{a_{j}}-\nabla^{a_{j+1}})f\cdot\rho_{j}\,(\nabla^{a_{j}%
}-\nabla^{a_{j+1}})\left(  \frac{\mathcal{R}}{\rho_{j}}\right)  \,d\mathbf{a}.
\label{e.Brian.int.by.parts.2}%
\end{align}
Using (\ref{secondTerm}) again, \eqref{e.Brian.int.by.parts.2} expands to
\begin{align}
&  -\int(\nabla^{a_{j}}-\nabla^{a_{j+1}})f\cdot\rho_{j}\,(\nabla^{a_{j}%
}-\nabla^{a_{j+1}})\left(  \frac{\mathcal{R}}{\rho_{j}}\right)  \,d\mathbf{a}%
\nonumber\\
=  &  \int(\nabla^{a_{j}}-\nabla^{a_{j+1}})f\cdot\left[  \frac{(\nabla
\rho_{t_{j-1}})(b_{j-1})}{\rho_{j-1}}+\frac{(\nabla\rho_{t_{j+1}})(b_{j+1}%
)}{\rho_{j+1}}\right]  \,d\mu_{\alpha,\mathbf{t}}.
\label{e.Brian.middle.terms.1}%
\end{align}

We now claim that the alternating sum of the terms in
(\ref{e.Brian.middle.terms.1}) is zero, assuming that $f$ has extended gauge
invariance. Consider, for example, the $j=1$ and $j=3$ terms, namely%
\begin{equation}
\int(\nabla^{a_{1}}-\nabla^{a_{2}})f\cdot\left[  \frac{(\nabla\rho_{t_{4}%
})(b_{4})}{\rho_{4}}+\frac{(\nabla\rho_{t_{2}})(b_{2})}{\rho_{2}}\right]
\,d\mu_{\alpha,\mathbf{t}} \label{j1}%
\end{equation}
and%
\begin{equation}
\int(\nabla^{a_{3}}-\nabla^{a_{4}})f\cdot\left[  \frac{(\nabla\rho_{t_{2}%
})(b_{2})}{\rho_{2}}+\frac{(\nabla\rho_{t_{4}})(b_{4})}{\rho_{4}}\right]
\,d\mu_{\alpha,\mathbf{t}}. \label{j2}%
\end{equation}
Since $f$ has extended gauge invariance, (\ref{e.inf.ext.gauge.inv}) tells us
that $\nabla^{a_{3}}f=-\nabla^{a_{1}}f$ and $\nabla^{a_{4}}f=-\nabla^{a_{2}%
}f.$ Thus, (\ref{j1}) and (\ref{j2}) cancel in the alterating sum. The $j=2$
and $j=4$ terms cancel similarly.

Thus, after integrating by parts in (\ref{firstStage}), only
\eqref{e.Brian.int.by.parts.1.5} contributes, giving
\begin{equation}
\sum_{j=1}^{4}(-1)^{j-1}\frac{\partial}{\partial t_{j}}\int f\,d\mu
_{\alpha,\mathbf{t}}=\frac{1}{8}\sum_{j=1}^{4}(-1)^{j-1}\int(\nabla^{a_{j}%
}-\nabla^{a_{j+1}})^{2}f\,d\mu_{\alpha,\mathbf{t}}. \label{e.penultimate}%
\end{equation}
The significance of this expression is that the derivatives $(\nabla^{a_{j}%
}-\nabla^{a_{j+1}})^{2}$ that were initially applied to $\rho_{j}$ are now
applied only to the function $f$ and not to any of the heat kernels.

Now, since $\nabla^{a_{j}}$ and $\nabla^{a_{j+1}}$ commute (because $a_{j}$
and $a_{j+1}$ are independent variables), we have
\[
(\nabla^{a_{j}}-\nabla^{a_{j+1}})^{2}=(\nabla^{a_{j}})^{2}+(\nabla^{a_{j+1}%
})^{2}-2\nabla^{a_{j}}\cdot\nabla^{a_{j+1}}.
\]
If we let $C_{j}=\displaystyle{\int(\nabla^{a_{j}})^{2}f\,d\mu_{\alpha
,\mathbf{t}}}$, then in (\ref{e.penultimate}), the $C_{j}$ terms will cancel
in pairs, leaving us with
\begin{align*}
&  \sum_{j=1}^{4}(-1)^{j-1}\frac{\partial}{\partial t_{j}}\int f\,d\mu\\
=  &  -\frac{1}{4}\int\left[  \nabla^{a_{1}}\cdot\nabla^{a_{2}}-\nabla^{a_{2}%
}\cdot\nabla^{a_{3}}+\nabla^{a_{3}}\cdot\nabla^{a_{4}}-\nabla^{a_{4}}%
\cdot\nabla^{a_{1}}\right]  f\,d\mu_{\alpha,\mathbf{t}}\\
=  &  -\int\nabla^{a_{1}}\cdot\nabla^{a_{2}}f\,d\mu,
\end{align*}
where we used \eqref{gradABforms} in the last equality. This is what we wanted
to prove. \hfill$\square$

\subsubsection{Second proof of Theorem \ref{localMM.thm}}

In our second proof, which is likely to be about as short as possible, we
begin by writing the density of $\mu_{\mathbf{\alpha},\mathbf{t}}$ as a
product of two terms: those corresponding to $(t_{1},t_{2})$ and those to
$(t_{3},t_{4})$:
\[
\mathcal{R}_{12}(\mathbf{a})=\rho_{t_{1}}(a_{2}^{-1}\alpha_{1}a_{1}%
)\rho_{t_{2}}(a_{3}^{-1}\alpha_{2}a_{2}),\quad\mathcal{R}_{34}(\mathbf{a}%
)=\rho_{t_{3}}(a_{4}^{-1}\alpha_{3}a_{3})\rho_{t_{4}}(a_{1}^{-1}\alpha
_{4}a_{4}).
\]
What is important is that $\mathcal{R}_{12}$ depends on $a_{1},a_{2},a_{3}$
but not $a_{4}$, while $\mathcal{R}_{34}$ depends on $a_{1},a_{3},a_{4}$, but
not $a_{2}$. Then
\[
d\mu_{\mathbf{\alpha},\mathbf{t}}=\mathcal{R}_{12}\mathcal{R}_{34}%
\,d\mathbf{a}.
\]
For the remainder of the proof, we write integrals of functions $g$ against
$d\mathbf{a}$ simply as $\int g$.

Now, using extended gauge invariance as in \eqref{gradABforms}, taking care to
commute partial derivatives, we may write
\[
\nabla^{a_{1}}\cdot\nabla^{a_{2}}f=\frac{1}{2}(\nabla^{a_{1}}-\nabla^{a_{3}%
})\cdot\nabla^{a_{2}}f.
\]
Then we integrate by parts once, and use the product rule.%
\begin{align*}
-\int\nabla^{a_{1}}\cdot\nabla^{a_{2}}f\,~d\mu_{\mathbf{\alpha},\mathbf{t}}
&  =-\frac{1}{2}\int[(\nabla^{a_{1}}-\nabla^{a_{3}})\cdot(\nabla^{a_{2}%
}f)]\,\mathcal{R}_{12}\mathcal{R}_{34}\\
&  =\frac{1}{2}\int[\nabla^{a_{2}}f\cdot(\nabla^{a_{1}}-\nabla^{a_{3}%
})(\mathcal{R}_{12}\mathcal{R}_{34})]\\
&  =\frac{1}{2}\int\mathcal{R}_{34}[\nabla^{a_{2}}f\cdot(\nabla^{a_{1}}%
-\nabla^{a_{3}})(\mathcal{R}_{12})]\\
&  +\frac{1}{2}\int\mathcal{R}_{12}[\nabla^{a_{2}}f\cdot(\nabla^{a_{1}}%
-\nabla^{a_{3}})(\mathcal{R}_{34})].
\end{align*}

We now use extended gauge invariance once more, in the second term, writing
$\nabla^{a_{2}}f=-\nabla^{a_{4}}f$, yielding
\[
\frac{1}{2}\int\mathcal{R}_{34}[\nabla^{a_{2}}f\cdot(\nabla^{a_{1}}%
-\nabla^{a_{3}})(\mathcal{R}_{12})]-\frac{1}{2}\int\mathcal{R}_{12}%
[\nabla^{a_{4}}f\cdot(\nabla^{a_{1}}-\nabla^{a_{3}})(\mathcal{R}_{34})].
\]
Since $\mathcal{R}_{12}$ does not depend on $a_{4}$, and $\mathcal{R}_{34}$
does not depend on $a_{2}$, we can integrate this by parts once more, and the
$\nabla^{a_{2}}$ and $\nabla^{a_{4}}$ derivatives only hit the already
differentiated factors. Thus,
\begin{align}
-\int\nabla^{a_{1}}\cdot\nabla^{a_{2}}f\,~d\mu_{\mathbf{\alpha},\mathbf{t}}
&  =-\frac{1}{2}\int f\mathcal{R}_{34}[\nabla^{a_{2}}\cdot(\nabla^{a_{1}%
}-\nabla^{a_{3}})\mathcal{R}_{12}]\nonumber\\
&  \;\;\;\,+\frac{1}{2}\int f\mathcal{R}_{12}[\nabla^{a_{4}}\cdot
(\nabla^{a_{1}}-\nabla^{a_{3}})\mathcal{R}_{34}]. \label{e.shortest.1}%
\end{align}

Finally, we compute the second derivatives. Recalling that $\mathcal{R}%
_{12}=\rho_{t_{1}}\rho_{t_{2}}$ and recalling the arguments of the heat
kernels from the definition of $\mu_{\mathbf{\alpha},\mathbf{t}}$, we have%
\begin{align*}
(\nabla^{a_{1}}-\nabla^{a_{3}})\mathcal{R}_{12}  &  =(\nabla^{a_{1}}%
-\nabla^{a_{3}})(\rho_{t_{1}}(a_{2}^{-1}\alpha_{1}a_{1})\rho_{t_{2}}%
(a_{3}^{-1}\alpha_{2}a_{2}))\\
&  =\rho_{t_{2}}(a_{3}^{-1}\alpha_{2}a_{2})(\nabla\rho_{t_{1}})(a_{2}%
^{-1}\alpha_{1}a_{1})\\
&  +\rho_{t_{1}}(a_{2}^{-1}\alpha_{1}a_{1})(\nabla\rho_{t_{2}})(a_{3}%
^{-1}\alpha_{2}a_{2}).
\end{align*}
Applying $\nabla^{a_{2}}$ then yields
\[
\nabla^{a_{2}}\cdot(\nabla^{a_{1}}-\nabla^{a_{3}})\mathcal{R}_{12}=\nabla
\rho_{t_{2}}\cdot\nabla\rho_{t_{1}}-\rho_{t_{2}}\mathrm{\Delta}\rho_{t_{1}%
}-\nabla\rho_{t_{2}}\cdot\nabla\rho_{t_{1}}+\rho_{t_{1}}\mathrm{\Delta}%
\rho_{t_{2}}.
\]
The first and third terms cancel, and we see that the first term on the
right-hand side of \eqref{e.shortest.1} is equal to%
\begin{align*}
-\frac{1}{2}\int f\mathcal{R}_{34}[\nabla^{a_{2}}\cdot(\nabla^{a_{1}}%
-\nabla^{a_{3}})\mathcal{R}_{12}]  &  =-\frac{1}{2}\int f\rho_{t_{3}}%
\rho_{t_{4}}(-\rho_{t_{2}}\mathrm{\Delta}\rho_{t_{1}}+\rho_{t_{1}%
}\mathrm{\Delta}\rho_{t_{2}})\\
&  =\left(  \frac{\partial}{\partial t_{1}}-\frac{\partial}{\partial t_{2}%
}\right)  \int f\,~d\mu_{\mathbf{\alpha},\mathbf{t}},
\end{align*}
where we have used the heat equation \eqref{e.heat.eq} in the second equality.

An entirely analogous computation shows that the second term on the right-hand
side of \eqref{e.shortest.1} is equal to $(\frac{\partial}{\partial t_{3}%
}-\frac{\partial}{\partial t_{4}})\int f\,d\mu_{\mathbf{\alpha},\mathbf{t}}$,
and adding these up gives the left-hand-side of Theorem \ref{MMab.thm},
concluding the proof. \hfill$\square$

\subsection{Additional examples of the abstract Makeenko--Migdal equation
\label{MMexamples.sec}}

We have noted that the Makeenko--Migdal equation in Theorem \ref{MMun.thm} is
a special case of the abstract Makeenko--Migdal equation. As L\'{e}vy has
noted \cite[Proposition 6.24]{Levy}, the abstract result can be specialized in
many interesting ways; we mention a few of these here. We now take $K=U(N),$
with metric normalized as in (\ref{HSN}).

First, suppose that $\mathbb{G}$ is a graph and $L_{1},\ldots,L_{r}$ are loops
traced out in $\mathbb{G}.$ Assume that $L_{r}$ has a simple crossing at $v$
and that none of the remaining loops passes through $v.$ Let $L_{r}^{\prime}$
and $L_{r}^{\prime\prime}$ be the splitting of $L_{r}$ at $v, $ as in the
statement of Theorem \ref{MMun.thm}. Then%
\begin{align*}
&  \left(  \frac{\partial}{\partial t_{1}}-\frac{\partial}{\partial t_{2}%
}+\frac{\partial}{\partial t_{3}}-\frac{\partial}{\partial t_{4}}\right)
\mathbb{E}\{\mathrm{tr}(\mathrm{hol}(L_{1}))\cdots\mathrm{tr}(\mathrm{hol}%
(L_{r-1}))\mathrm{tr}(\mathrm{hol}(L_{r}))\}\\
&  =\mathbb{E}\{\mathrm{tr}(\mathrm{hol}(L_{1}))\cdots\mathrm{tr}%
(\mathrm{hol}(L_{r-1}))\mathrm{tr}(\mathrm{hol}(L_{r}^{\prime}))\mathrm{tr}%
(\mathrm{hol}(L_{r}^{\prime\prime}))\}.
\end{align*}
The derivation of this example from Theorem \ref{MMab.thm} is precisely the
same as in (\ref{MMholonomy}); the additional loop $L_{1},\ldots,L_{r-1}$
simply tag along for the ride.

\begin{figure}[ptb]
\centering
\includegraphics[
height=1.9303in,
width=2.5564in
]{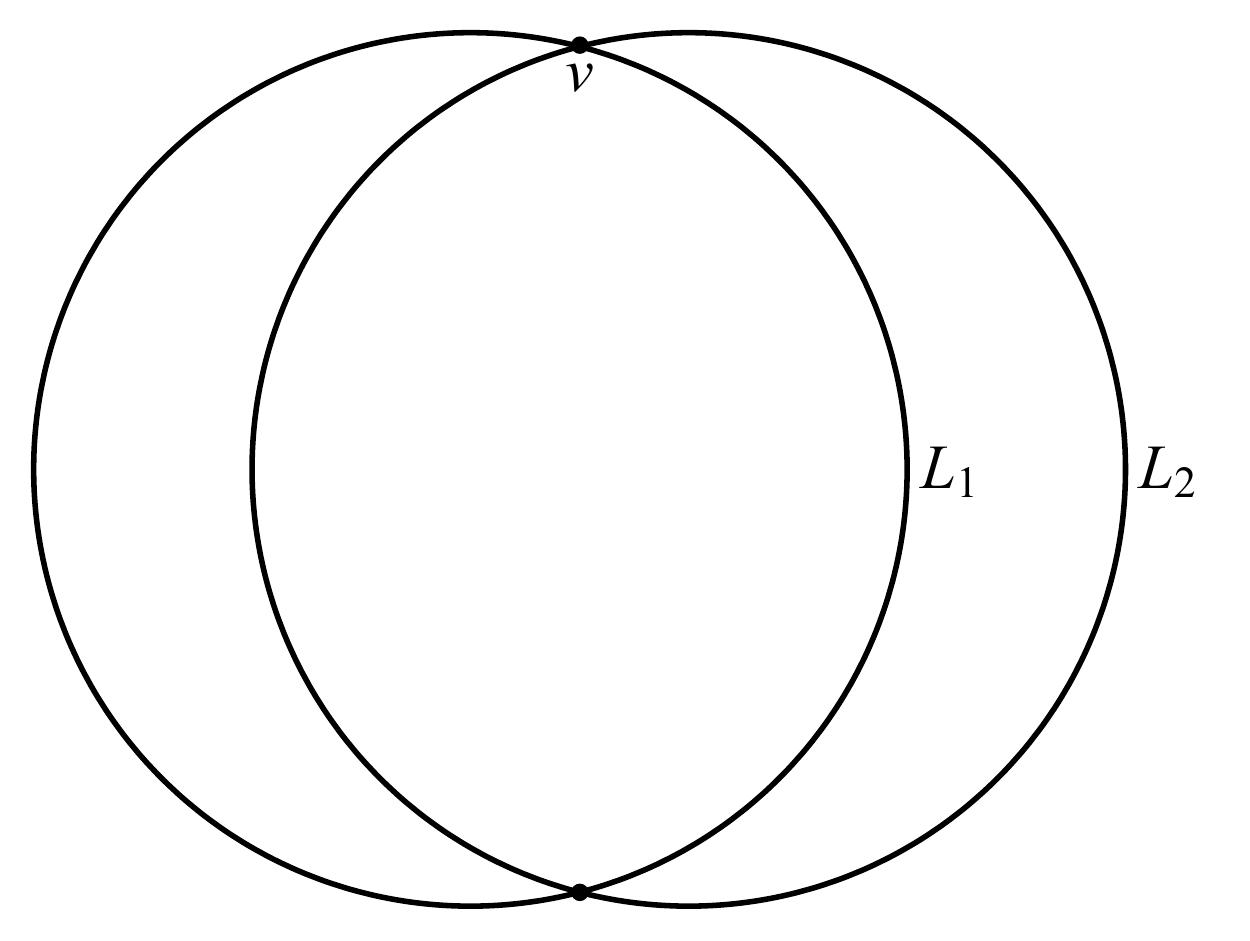}\caption{Two loops passing through $v$}%
\label{twoloops.fig}%
\end{figure}

Second, suppose $L_{1}$ and $L_{2}$ are two loops traced out in a graph
$\mathbb{G}$ and suppose $L_{1}$ starts at $v$, goes out along $e_{1}$, and
then eventually returns to $v$ along $e_{3}^{-1}$, but otherwise does not pass
through any of $e_{1},\ldots,e_{4}$. Suppose $L_{2}$ starts at $v$ goes out
along $e_{2}$ and then eventually returns to $v$ along $e_{4}^{-1}$, but
otherwise does not pass through any of $e_{1},\ldots,e_{4}.$ (See Figure
\ref{twoloops.fig} for a simple example.) We will shortly verify that
\begin{align}
&  \left(  \frac{\partial}{\partial t_{1}}-\frac{\partial}{\partial t_{2}%
}+\frac{\partial}{\partial t_{3}}-\frac{\partial}{\partial t_{4}}\right)
\mathbb{E}\left[  \mathrm{tr}(\mathrm{hol}(L_{1}))\mathrm{tr}(\mathrm{hol}%
(L_{2}))\right] \nonumber\\
&  =\frac{1}{N^{2}}\mathbb{E}\left[  \mathrm{tr}\{\mathrm{hol}(L_{1}%
)\mathrm{hol}(L_{2})\}\right]  . \label{twoLoops}%
\end{align}
For this example, we note that the integrand on the left-hand side of
(\ref{twoLoops}) has the form
\[
f(\mathbf{x})=\mathrm{tr}(a_{3}^{-1}\alpha a_{1})\mathrm{tr}(a_{4}^{-1}\beta
a_{2})
\]
where $\alpha$ and $\beta$ are words in the $\mathbf{b}$ variables. This
function has extended gauge invariance at $v$, and we find that%
\begin{align*}
(\nabla^{a_{1}}\cdot\nabla^{a_{2}}f)(\mathbf{x})  &  =\sum_{X}\mathrm{tr}%
(a_{3}^{-1}\alpha a_{1}X)\mathrm{tr}(a_{4}^{-1}\beta a_{2}X)\\
&  =-\frac{1}{N^{2}}\mathrm{tr}(a_{3}^{-1}\alpha a_{1}a_{4}^{-1}\beta a_{2})\\
&  =-\frac{1}{N^{2}}\mathrm{tr}[\mathrm{hol}(L_{1})\mathrm{hol}(L_{2})],
\end{align*}
where in the second equality, we have used a simple identity (e.g., the last
line of Proposition 3.1 in \cite{DHK}). Then (\ref{twoLoops}) follows from
Theorem \ref{MMab.thm}.

\section{A proof using loop variables\label{loop.sec}}

In Section 7 of \cite{Dahl}, Dahlqvist gave a proof of the Makeenko--Migdal
equation for $U(N)$ (Theorem \ref{MMun.thm}) using \textquotedblleft
loop\textquotedblright\ or \textquotedblleft lasso\textquotedblright%
\ variables. This proof stands in contrast to the proof in \cite{Levy} of the
abstract Makeenko--Migdal equation (which implies Theorem \ref{MMun.thm})
using \textquotedblleft edge\textquotedblright\ variables. Like L\'{e}vy's
proof using edge variables, Dahlqvist's proof is based on a formula for the
derivative with respect to an individual time variable: a formula which
contains a large number of terms that must cancel upon taking the alternating
sum. We give a new loop-based proof in which we work with the alternating sum
from the beginning and obtain the necessary cancellations without ever
encountering all the terms arising in \cite{Dahl}. Our loop-based proof
actually establishes the abstract Makeenko--Migdal (Theorem \ref{MMab.thm})
for functions that are gauge invariant, in addition to having extended gauge
invariance at the crossing in question. This result contains, as a special
case, the Makeenko--Migdal equation for $U(N).$

The goal of this section is to give a different proof of the following special
case of the abstract Makeenko--Migdal equation.

\begin{theorem}
\label{mmeLoop.thm}Let $v$ be a vertex of $\mathbb{G}$ with four incident
edges. Assume that $f$ is gauge invariant in the sense of Definition
\ref{discreteGaugeInv.def} and that $f$ has extended gauge invariance at $v.$
Then the abstract Makeenko--Migdal equation for $f$ at $v$ holds, as in
Theorem \ref{MMab.thm}.
\end{theorem}

If the four edges at $v$ are not distinct, extended gauge invariance should be
interpreted as in Section \ref{interpret.sec}. This result implies, in
particular, the usual Makeenko--Migdal equation for the trace in $U(N)$ of the
holonomy around a loop with a simple crossing at $v.$

\subsection{The loop variables}

It is well known that the fundamental group of any graph is free. Given a
planar graph $\mathbb{G}$, a fixed vertex $v$ of $\mathbb{G}$, and a spanning
tree $\mathbb{T}$ in $\mathbb{G}$, L\'{e}vy gives a particular set of free
generators for $\pi_{1}(\mathbb{G})$, which we refer to as loops or lassos.
The generators are in one-to-one correspondence with the bounded faces of
$\mathbb{G}$, and are constructed as follows. First, the choice of
$\mathbb{T}$ determines, for each bounded face $F$ of $\mathbb{G},$ a
distinguished edge $e_{F}$ on the boundary of $F$ that bounds $F$ positively.
Roughly, we travel from $F$ to the unbounded face of $\mathbb{G}$ by crossing
only edges not in $\mathbb{T},$ and $e$ is the first such edge crossed. To be
more precise, we look for a sequence of faces $F_{1},\ldots,F_{n}$ with
$F_{1}=F$ and $F_{n}$ being the unbounded face, and a sequence of edges
$e_{1},\ldots,e_{n-1}$ where $e_{i}$ is not in $\mathbb{T},$ lies on the
boundary of both $F_{i}$ and $F_{i+1},$ and bounds $F_{i}$ positively. The
$F_{i}$'s and $e_{i}$'s exist and are unique, as long as we do not immediately
recross an edge we just crossed; that is, $e_{i+1}$ should not be the inverse
of $e_{i}$. (See Section 4.3 of \cite{Levy} for more information.) For a given
bounded face $F,$ we set $e_{F}=e_{1}.$

L\'{e}vy then associates to each bounded face $F$ of $\mathbb{G}$ a loop
$l_{F}$, as follows. Let $v_{F}$ be the initial vertex of $e_{F}.$ We then
start at the base point $v,$ proceed to $v_{F}$ along a path $p$ in
$\mathbb{T},$ travel around $F$ in the positive direction (beginning with the
edge $e_{F}$), and then return to $v$ along the inverse of $p.$ The loop
$l_{F}$ is then the reduced loop obtained by removing backtracks from the
just-described loop. According to Proposition 4.2 of \cite{Levy}, the $l_{F}%
$'s form a set of free generators for $\pi_{1}(\mathbb{G}).$ (The well-known
general procedure for constructing free generators of the fundamental group of
a finite graph associates a free generator to each undirected edge of
$\mathbb{G}\setminus\mathbb{T}.$ L\'{e}vy's construction refines this
procedure in the planar case by making a one-to-one correspondence between the
undirected edges of $\mathbb{G}\setminus\mathbb{T}$ and the bounded faces of
$\mathbb{G}.$)

We now introduce the loop variables, which are simply the products of the edge
variables associated to the edges in the just-defined loops. The loop
variables are almost the same as the holonomy variables $h_{i}$ entering into
Driver's formula (\ref{DriversForm}), except that they contain a
\textquotedblleft tail\textquotedblright\ representing the path $p$ in the
previous paragraph. (Since $\rho_{t}$ is conjugation invariant, the tail may
be omitted from the heat kernel.)

In Theorem \ref{mmeLoop.thm}, we do not make any genericity assumptions on
$\mathbb{G}.$ In the proof, however, it is convenient to assume at first---as
in our other proofs---that $\mathbb{G}$ is generic at $v,$ meaning that the
four edges emanating from $v$ are distinct, and that the four adjacent faces
are distinct and bounded. In addition, it is convenient to assume that it is
possible to choose a spanning tree $\mathbb{T}$ for $\mathbb{G}$ in such a way
that the loops $L_{i}$ associated to the adjacent faces have the following
form:
\begin{equation}
L_{i}=\partial F_{i}=e_{i}A_{i}e_{i+1}^{-1},\quad i=1,2,3,4,
\label{adjacentLoopAssumption}%
\end{equation}
where $A_{i}$ is a word in edges not belonging to $\{e_{1},e_{2},e_{3}%
,e_{4}\}.$ In the current section, we prove Theorem \ref{mmeLoop.thm} under
both of these assumptions on $\mathbb{G}.$ In Section \ref{generic.sec}, these
assumptions will be lifted.

If the adjacent loops have the form in (\ref{adjacentLoopAssumption}), then
since parallel transport is order-reversing, the corresponding loop
\textit{variables} (with values in $K$) will have the form%
\begin{equation}
\ell_{i}=a_{i+1}^{-1}\alpha_{i}a_{i},\quad i=1,2,3,4, \label{adjacentLoops}%
\end{equation}
where $\alpha_{i}$ is a word in the $\mathbf{b}$ variables (that is, the edge
variables not belonging to $\{a_{1},a_{2},a_{3},a_{4}\}$).

\begin{figure}[ptb]
\centering
\includegraphics[
height=2.5564in,
width=2.5564in
]{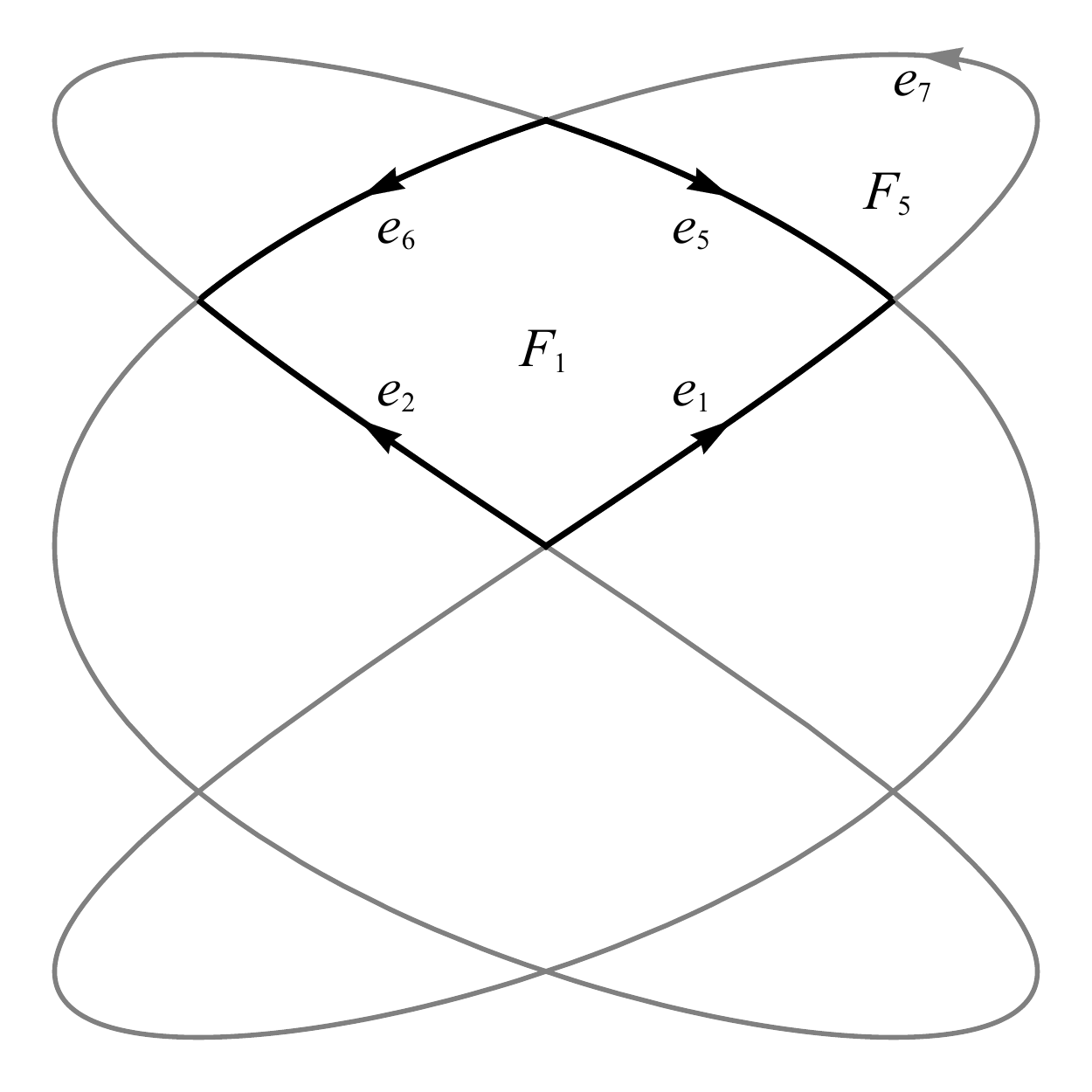}\caption{The loop $L_{1}$ associated to $F_{1}$ is $e_{1}%
e_{5}^{-1}e_{6}e_{2}^{-1}$}%
\label{loopvar1.fig}%
\end{figure}

Meanwhile, for the loop $L_{j}$ associated to any bounded face \textit{other}
than $F_{1},$ $F_{2},$ $F_{3},$ and $F_{4},$ the associated loop will have the
form $L_{j}=e_{i_{j}}B_{j}e_{i_{j}}^{-1},$ where $e_{i_{j}}\in\{e_{1}%
,e_{2},e_{3},e_{4}\}$ is the first edge traversed by $L_{j}$ and where $B_{j}$
is a word in edge variables not in $\{e_{1},e_{2},e_{3},e_{4}\}.$ Thus, the
corresponding loop \textit{variable} will have the form%
\begin{equation}
\ell_{j}=a_{i_{j}}^{-1}\beta_{j}a_{i_{j}},\quad j\geq5,
\label{nonAdjacentLoops}%
\end{equation}
where $\beta_{j}$ is a word in the $\mathbf{b}$ variables. In Figures
\ref{loopvar1.fig} and \ref{loopvar2.fig}, for example, if the spanning tree
is chosen to include the edges $e_{1}$, $e_{2}$, and $e_{6}$, then the
distinguished edges associated to $F_{1}$ and $F_{5}$ in L\'{e}vy's procedure
will be $e_{5}^{-1} $ and $e_{7}$, respectively, and the associated loops
$L_{1}$ and $L_{5}$ will be as indicated in the figures.

\begin{figure}[ptb]
\centering
\includegraphics[
height=2.5564in,
width=2.5564in
]{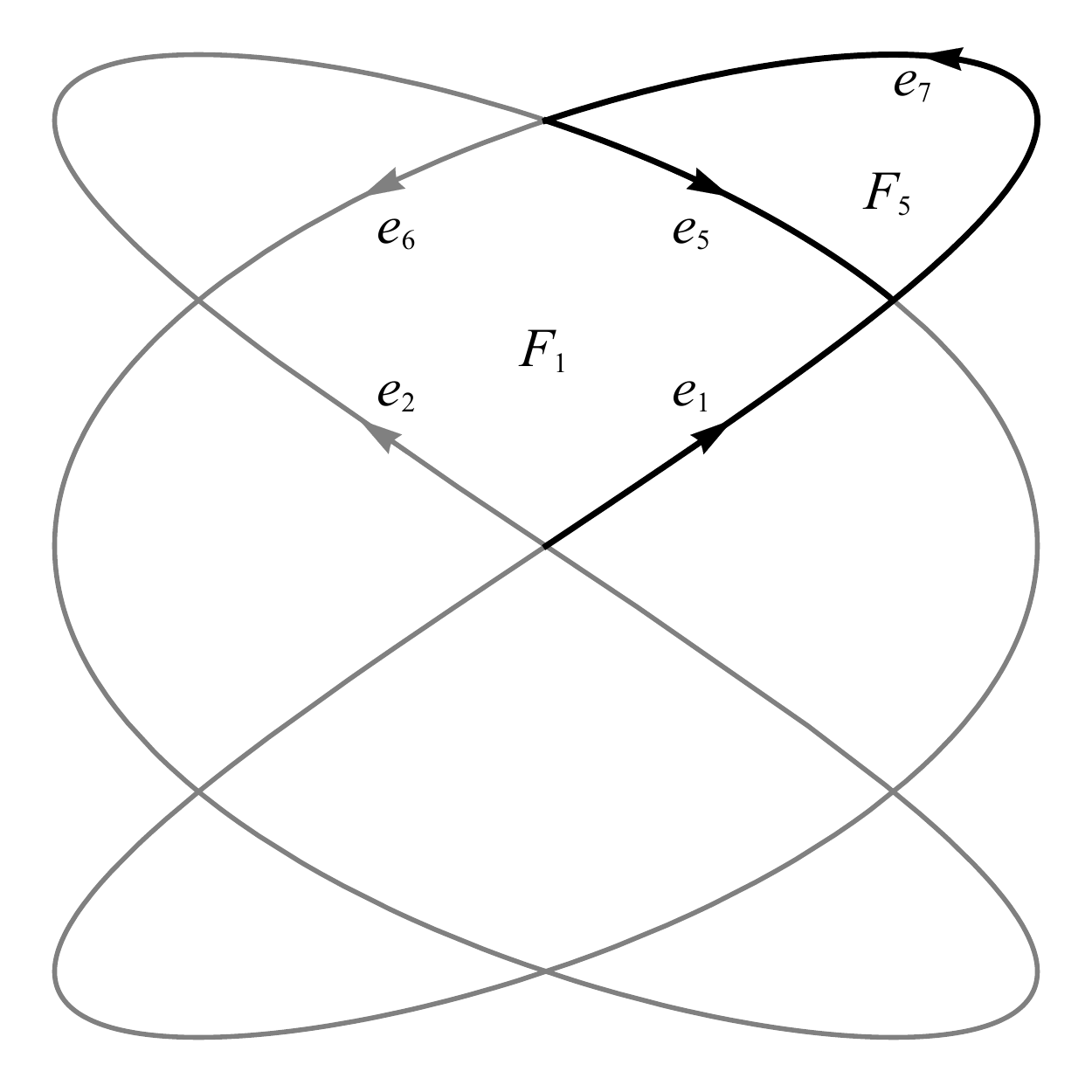}\caption{The loop $L_{5}$ associated to $F_{5}$ is $e_{1}%
e_{7}e_{5}e_{1}^{-1} $}%
\label{loopvar2.fig}%
\end{figure}

If $n$ is the number of (non-oriented) edges of $\mathbb{G}$ and $m$ is the
number of bounded faces, the assignments in (\ref{adjacentLoops}) and
(\ref{nonAdjacentLoops}) define a smooth map $\mathrm{\Gamma}:K^{n}\rightarrow
K^{m},$ sending the edge variables to the loop variables. Typically, there
will be fewer loop variables than edge variables. Thus, not every function of
the edge variables (whether or not the function has extended gauge invariance
at a particular vertex $v$) will be expressible as a function of the loop
variables. On the other hand, since the basic loops generate the fundamental
group of $\mathbb{G},$ the trace of the holonomy around any closed curve in
$\mathbb{G}$ will be expressible as a function of the loop variables. More
generally, according to Lemma 2.1.5 of \cite{Levy1}, if $f$ is gauge invariant
(Definition \ref{discreteGaugeInv.def}), then $f$ can be expressed as a
function of the loop variables. (More precisely, the cited result of L\'{e}vy
shows that a gauge-invariant function can be expressed in terms of
\textit{certain} loop variables; since the $L_{j}$'s generate $\pi
_{1}(\mathbb{G}),$ those loop variables can then be expressed in terms of the
$\ell_{j}$'s.)

Meanwhile, L\'{e}vy has shown, using Driver's formula \eqref{DriversForm} for
the Yang--Mills measure for $\mathbb{G}$, that the loop variables are
independent and heat kernel distributed. (See Proposition 4.4 in \cite{Levy}.)
That is to say: the push-forward of the measure $\mu$ under the map
$\mathrm{\Gamma}$ is simply the product of heat kernel measures with time
parameters equal to the areas of the bounded faces. (This independence result
does not hold for surfaces other than the plane; thus, our proof using the
loop variables does not extend to general surfaces.) If $\tilde{\mu}$ refers
to this pushed forward measure, the measure-theoretic change of variables
theorem says that if $f=g\circ\mathrm{\Gamma},$ then
\begin{equation}
\int_{K^{n}}f~d\mu=\int_{K^{m}}g~d\tilde{\mu}. \label{changeOfMeasure}%
\end{equation}

We now consider how changes in the four \textquotedblleft
adjacent\textquotedblright\ edge variables affect the loop variables. If we
change, say, $a_{1}$ to $a_{1}x,$ we can read off the corresponding change in
the adjacent loop variables (as in (\ref{adjacentLoops})) as%
\[
(\ell_{1},\ell_{2},\ell_{3},\ell_{4})\mapsto(\ell_{1}x,\ell_{2},\ell
_{3},x^{-1}\ell_{4}).
\]
Meanwhile, each nonadjacent loop variable $\ell_{j}$ with $j\geq5$ (as in
(\ref{nonAdjacentLoops})) will either be conjugated by $x$ or unchanged,
depending on whether $L_{j}$ goes out along $e_{1}$ or along $e_{2},$ $e_{3},
$ or $e_{4}.$ Similar transformation rules hold for changes in $a_{2}, $
$a_{3},$ and $a_{4}.$

In particular, if we make the substitution
\[
(a_{1},a_{2},a_{3},a_{4})\mapsto(a_{1}x,a_{2}y,a_{3}x,a_{4}y)
\]
we have the following substitutions for the adjacent loop variables:%
\begin{equation}
(\ell_{1},\ell_{2},\ell_{3},\ell_{4})\mapsto(y^{-1}\ell_{1}x,x^{-1}\ell
_{2}y,y^{-1}\ell_{3}x,x^{-1}\ell_{4}y), \label{extendLoop1}%
\end{equation}
whereas each loop $\ell_{j}$ with $j\geq5$ changes either as%
\begin{equation}
\ell_{j}\mapsto x^{-1}\ell_{j}x \label{extendLoop2}%
\end{equation}
or as
\begin{equation}
\ell_{j}\mapsto y^{-1}\ell_{j}y \label{extendLoop3}%
\end{equation}
depending on the first outgoing edge traversed by the loop $L_{j}$.

The above calculation motivates the following definition.

\begin{definition}
\label{extendedInvariance.def}We say that a function on $K^{m}$ has
\emph{extended gauge invariance} at $v$ if it is invariant under every
transformation of the sort in \eqref{extendLoop1}, \eqref{extendLoop2}, and \eqref{extendLoop3}.
\end{definition}

Since changes in the variables $a_{1},$ $a_{2},$ $a_{3},$ and $a_{4}$
translate into simple changes in the loop variables, we can translate the
differential operators $\nabla^{a_{i}}\cdot\nabla^{a_{i+1}}$ on $K^{n}$ into
differential operators on $K^{m}.$ In a slight abuse of notation, we will
continue to refer to the operators on $K^{m}$ as $\nabla^{a_{i}}\cdot
\nabla^{a_{i+1}}.$

\begin{definition}
\label{gradabLoop.def}Suppose $g:K^{m}\rightarrow\mathbb{C}$ is a smooth
function of the loop variables. Then $\nabla^{a_{i}}\cdot\nabla^{a_{j}}%
g:K^{m}\rightarrow\mathbb{C},$ with $i\neq j$ in $\{1,2,3,4\},$ is the
function computed as follows. Let $\boldsymbol{\ell}=(\ell_{i})_{i=1}^{m}$ be
the loop variables defined in \eqref{adjacentLoops} and
\eqref{nonAdjacentLoops}. Define a parametrized surface $\boldsymbol{\ell
}(s,t)$ in $K^{m}$ by replacing $a_{i}$ and $a_{j}$ with $a_{i}\mapsto
a_{i}e^{sX}$ and $a_{j}\mapsto a_{j}e^{tX}$ in those two equations. We then
set%
\[
(\nabla^{a_{i}}\cdot\nabla^{a_{j}}g)(\boldsymbol{\ell})=\left.  \sum_{X}%
\frac{\partial^{2}}{\partial s\partial t}g(\boldsymbol{\ell}(s,t))\right\vert
_{s=t=0}.
\]

\end{definition}

The notation is suggestive because if $\mathrm{\Gamma}$ is the map from the
edge variables to the loop variables, we have $\nabla^{a_{i}}\cdot
\nabla^{a_{j}}(g\circ\mathrm{\Gamma})=(\nabla^{a_{i}}\cdot\nabla^{a_{j}%
}g)\circ\mathrm{\Gamma}.$ We will compute these operators in the next subsection.

\subsection{The proof}

We now prove Theorem \ref{mmeLoop.thm}, under the following assumptions on
$\mathbb{G}$: First, the edges $e_{1},\ldots,e_{4}$ at $v$ are distinct,
second, the faces $F_{1},\ldots,F_{4}$ adjacent to $v$ are distinct and
bounded, and third, it is possible to choose a spanning tree $\mathbb{T}$ for
$\mathbb{G}$ in such a way that the loops associated to $F_{1},\ldots,F_{4}$
have the form given in (\ref{adjacentLoopAssumption}). All of these
assumptions are lifted in Section \ref{generic.sec}.

Suppose $f$ satisfies the hypotheses of Theorem \ref{mmeLoop.thm}. As we have
noted, since $f$ is gauge invariant, it can be expressed as $f=g\circ
\mathrm{\Gamma},$ where $g$ is a function of the loop variables. Since, also,
$f $ has extended gauge invariance at $v,$ the function $g$ has extended gauge
invariance at $v$ in the sense of Definition \ref{extendedInvariance.def}. To
prove the theorem, it suffices to show that%

\begin{equation}
\left(  \frac{\partial}{\partial t_{1}}-\frac{\partial}{\partial t_{2}}%
+\frac{\partial}{\partial t_{3}}-\frac{\partial}{\partial t_{4}}\right)
\int_{K^{m}}g~d\tilde{\mu}=-\int_{K^{m}}\nabla^{a_{1}}\cdot\nabla^{a_{2}%
}g~d\tilde{\mu}, \label{loopVersion}%
\end{equation}
where $\nabla^{a_{1}}\cdot\nabla^{a_{2}}$ is the differential operator in
Definition \ref{gradabLoop.def}.

The advantage of working with the loop variables is that since each heat
kernel is evaluated on a separate loop variable, when we differentiate and
integrate by parts, none of the derivatives hits on any other heat kernel, but
only on the function $g.$ Thus,%
\begin{align}
&  \left(  \frac{\partial}{\partial t_{1}}-\frac{\partial}{\partial t_{2}%
}+\frac{\partial}{\partial t_{3}}-\frac{\partial}{\partial t_{4}}\right)
\int_{K^{m}}g~d\tilde{\mu}\nonumber\\
&  =\frac{1}{2}\int_{K^{m}}\left(  \mathrm{\Delta}_{\ell_{1}}-\mathrm{\Delta
}_{\ell_{2}}+\mathrm{\Delta}_{\ell_{3}}-\mathrm{\Delta}_{\ell_{4}}\right)
g~d\tilde{\mu}. \label{LapTerm}%
\end{align}
On the other hand, even if $g$ has extended gauge invariance, the integrand on
the right-hand side of (\ref{LapTerm}) will contain many other terms besides
the one we want. (Compare Section \ref{example.sec}.) We will show that these
unwanted terms cancel out after integration.

We now consider the right-hand side of the loop form of the Makeenko--Migdal
equation, as in (\ref{loopVersion}). A key point will be to exploit the
invariance of the measure $\tilde{\mu}$ under conjugation in each variable. We
consider both a left-invariant gradient $\nabla^{L}$ and right-invariant
gradient $\nabla^{R}$ in the loop variables, similar to Definition
\ref{gradient.def}:%
\begin{align*}
(\nabla^{L}g)(\ell)  &  =\sum_{X}\left(  \left.  \frac{d}{ds}g(\ell
e^{sX})\right\vert _{s=0}\right)  ~X\\
(\nabla^{R}g)(\ell)  &  =\sum_{X}\left(  \left.  \frac{d}{ds}g(e^{sX}%
\ell)\right\vert _{s=0}\right)  ~X.
\end{align*}

\begin{lemma}
\label{adInv.lem}If $g$ is any smooth function on $K^{m},$ we have%
\begin{align*}
\int_{K^{m}}\nabla^{\ell_{j},L}\cdot\nabla^{\ell_{k},L}g~d\tilde{\mu}  &
=\int_{K^{m}}\nabla^{\ell_{j},L}\cdot\nabla^{\ell_{k},R}g~d\tilde{\mu}\\
&  =\int_{K^{m}}\nabla^{\ell_{j},R}\cdot\nabla^{\ell_{k},L}g~d\tilde{\mu}%
=\int_{K^{m}}\nabla^{\ell_{j},R}\cdot\nabla^{\ell_{k},R}g~d\tilde{\mu},
\end{align*}
where $\nabla^{\ell_{j},L}$ and $\nabla^{\ell_{j},R}$ denote the
left-invariant and right-invariant gradients in the variable $\ell_{j}$,
respectively, with the other variables fixed.
\end{lemma}

\begin{proof}
The invariance of each heat kernel under conjugation (cf.\ \eqref{e.conj.inv})
tells us that for any $j,$ we have%
\[
\int_{K^{m}}g(\ell_{1},\ldots,\ell_{j}x,\ldots,\ell_{m})~d\tilde{\mu}%
=\int_{K^{m}}g(\ell_{1},\ldots,x\ell_{j},\ldots,\ell_{m})~d\tilde{\mu}.
\]
If follows that%
\begin{align*}
&  \left.  \frac{\partial^{2}}{\partial s\partial t}\int_{K^{m}}g(\ell
_{1},\ldots,\ell_{j}e^{tX},\ldots,\ell_{k}e^{sX},\ldots,\ell_{m})~d\tilde{\mu
}\right\vert _{s=t=0}\\
=  &  \left.  \frac{\partial^{2}}{\partial s\partial t}\int_{K^{m}}g(\ell
_{1},\ldots,e^{tX}\ell_{j},\ldots,e^{sX}\ell_{k},\ldots,\ell_{m})~d\tilde{\mu
}\right\vert _{s=t=0},
\end{align*}
from which it follows that $\int_{K^{m}}\nabla^{\ell_{j},L}\cdot\nabla
^{\ell_{k},L}g~d\tilde{\mu}$ coincides with $\int_{K^{m}}\nabla^{\ell_{j}%
,R}\cdot\nabla^{\ell_{k},R}g~d\tilde{\mu}.$ (The argument works equally well
whether $j=k$ or $j\neq k.$) All the other claimed equalities follow by
analogous arguments.
\end{proof}

In what follows, we will make repeated use, usually without mention, of Lemma
\ref{adInv.lem}. It is convenient, in this context, to use the notation%
\[
f\cong g
\]
to indicate that $f$ and $g$ have the same integral.

We now work out more explicitly the differential operators $\nabla^{a_{i}%
}\cdot\nabla^{a_{j}}$ in Definition \ref{gradabLoop.def}. If we make the
substitutions $a_{1}\mapsto a_{1}e^{sX}$ and $a_{2}\mapsto a_{2}e^{tX},$ the
loop variables change as in (\ref{extendLoop1}), (\ref{extendLoop2}), and
(\ref{extendLoop3}). In particular, under such substitutions, each $\ell_{j}$
with $j\geq5$ merely gets conjugated or not changed at all. Let $\mathbf{m}$
denote the tuple of variables $\ell_{5},\ldots,\ell_{m}$ and let
$\mathbf{m}^{\prime}$ denote the new value of the these variables after
changing $a_{1}$ and $a_{2}$ as above. Then, by the conjugation invariance of
the measure, we have
\begin{align*}
\int_{K^{m}}\nabla^{a_{1}}\cdot\nabla^{a_{2}}g~d\tilde{\mu}  &  =\sum
_{X}\left.  \frac{\partial^{2}}{\partial s\partial t}\int_{K^{m}}\!\!g(e^{-tX}%
\ell_{1}e^{sX},\ell_{2}e^{tX},\ell_{3},e^{-sX}\ell_{4},\mathbf{m}^{\prime
})d\tilde{\mu}\right\vert _{s=t=0}\\
&  =\sum_{X}\left.  \frac{\partial^{2}}{\partial s\partial t}\int_{K^{m}}
\!\!g(e^{(s-t)X}\ell_{1},e^{tX}\ell_{2},\ell_{3},e^{-sX}\ell_{4},\mathbf{m}%
)d\tilde{\mu}\right\vert _{s=t=0}\\
&  =\int_{K^{m}}(-\mathrm{\Delta}_{\ell_{1}}+\nabla^{\ell_{1}}\cdot
\nabla^{\ell_{2}}+\nabla^{\ell_{1}}\cdot\nabla^{\ell_{4}}-\nabla^{\ell_{2}%
}\cdot\nabla^{\ell_{4}})g d\tilde{\mu}.
\end{align*}
(Recall that $\mathbf{m}^{\prime}$ differs from $\mathbf{m}$ only by
conjugations in some of the variables, which does not affect the value of the
integral.) Note that in light of Lemma \ref{adInv.lem}, we do not have to
specify whether the gradients are left-invariant or right-invariant.

After integrating and using the conjugation invariance of the measure, we are
left with a \textquotedblleft local\textquotedblright\ formula for
$\int_{K^{m}}\nabla^{a_{1}}\cdot\nabla^{a_{2}}g~d\tilde{\mu},$ that is, one in
which only derivatives in the variables $\ell_{1},$ $\ell_{2},$ $\ell_{3}$,
and $\ell_{4}$ enter. Since (\ref{LapTerm}) is also local in this sense, there
is no need to consider derivatives in any variables not belonging to
$\{\ell_{1},\ell_{2},\ell_{3},\ell_{4}\}.$

Using similar calculations for the other pairs of cyclically adjacent
variables, we have%
\begin{align*}
\nabla^{a_{1}}\cdot\nabla^{a_{2}}g  &  \cong(-\mathrm{\Delta}_{\ell_{1}%
}+\nabla^{\ell_{1}}\cdot\nabla^{\ell_{2}}+\nabla^{\ell_{1}}\cdot\nabla
^{\ell_{4}}-\nabla^{\ell_{2}}\cdot\nabla^{\ell_{4}})g\\
\nabla^{a_{2}}\cdot\nabla^{a_{3}}g  &  \cong(-\mathrm{\Delta}_{\ell_{2}%
}+\nabla^{\ell_{2}}\cdot\nabla^{\ell_{3}}+\nabla^{\ell_{1}}\cdot\nabla
^{\ell_{2}}-\nabla^{\ell_{1}}\cdot\nabla^{\ell_{3}})g\\
\nabla^{a_{3}}\cdot\nabla^{a_{4}}g  &  \cong(-\mathrm{\Delta}_{\ell_{3}%
}+\nabla^{\ell_{3}}\cdot\nabla^{\ell_{4}}+\nabla^{\ell_{2}}\cdot\nabla
^{\ell_{3}}-\nabla^{\ell_{2}}\cdot\nabla^{\ell_{4}})g\\
\nabla^{a_{4}}\cdot\nabla^{a_{1}}g  &  \cong(-\mathrm{\Delta}_{\ell_{4}%
}+\nabla^{\ell_{1}}\cdot\nabla^{\ell_{4}}+\nabla^{\ell_{3}}\cdot\nabla
^{\ell_{4}}-\nabla^{\ell_{1}}\cdot\nabla^{\ell_{3}})g.
\end{align*}
Now, if $g$ has extended gauge invariance, each of these terms reduces to one
of $\pm\nabla^{a_{1}}\cdot\nabla^{a_{2}}g.$ Hence, we may take an alternating
sum and divide by 4 to obtain%
\begin{align*}
\nabla^{a_{1}}\cdot\nabla^{a_{2}}g  &  \cong\left(  -\frac{1}{2}%
\mathrm{\Delta} _{\ell_{1}}+\frac{1}{2}\mathrm{\Delta}_{\ell_{2}}-\frac{1}%
{2}\mathrm{\Delta}_{\ell_{3}}+\frac{1}{2}\mathrm{\Delta}_{\ell_{4}}\right)
g\\
&  +\left(  \frac{1}{4}\mathrm{\Delta}_{\ell_{1}}-\frac{1}{4}\mathrm{\Delta
}_{\ell_{2}}+\frac{1}{4}\mathrm{\Delta}_{\ell_{3}}-\frac{1}{4}\mathrm{\Delta
}_{\ell_{4}}\right)  g\\
&  +\frac{1}{2}\left(  \nabla^{\ell_{1}}\cdot\nabla^{\ell_{3}}-\nabla
^{\ell_{2}}\cdot\nabla^{\ell_{4}}\right)  g,
\end{align*}
where we have written the \textquotedblleft correct\textquotedblright%
\ Laplacian term (as in (\ref{LapTerm})) on the first line. To establish the
Makeenko--Migdal equation, we need to prove that the last two lines disappear
after integration:%
\begin{equation}
\int_{K^{m}}\left(  \left[  \mathrm{\Delta}_{\ell_{1}}-\mathrm{\Delta}%
_{\ell_{2}}+\mathrm{\Delta} _{\ell_{3}}-\mathrm{\Delta}_{\ell_{4}}%
+2\nabla^{\ell_{1}}\cdot\nabla^{\ell_{3}}-2\nabla^{\ell_{2}}\cdot\nabla
^{\ell_{4}}\right]  g\right)  ~d\tilde{\mu}=0, \label{integratedIdentity}%
\end{equation}
whenever $g$ has extended gauge invariance at $v.$

To establish \eqref{integratedIdentity}, we recall that extended gauge
invariance means invariance under the transformations in \eqref{extendLoop1},
(\ref{extendLoop2}), and (\ref{extendLoop3}). Applying these transformations
with $x$ equal to $e^{tX}$ and $y$ equal to the identity and differentiating
shows that%
\begin{equation}
0=(\nabla^{\ell_{1},L}-\nabla^{\ell_{2},R}+\nabla^{\ell_{3},L}-\nabla
^{\ell_{4},R})g+\sum_{j\in I}(\nabla^{\ell_{j},L}-\nabla^{\ell_{j},R})g,
\label{adInvId}%
\end{equation}
where $I$ refers to the set of indices $j\geq5$ for which the loop goes out
from the basepoint along $a_{1}$ or $a_{3}.$ We now apply the operator
$\nabla^{\ell_{1},L}+\nabla^{\ell_{2},L}+\nabla^{\ell_{3},L}+\nabla^{\ell
_{4},L}$ to both sides of (\ref{adInvId}), integrate against $\tilde{\mu},$
and use Lemma \ref{adInv.lem}. All terms involving derivatives with respect to
$\ell_{j},$ $j\geq5,$ will drop out, and we do not have to specify whether the
remaining derivatives are left-invariant or right-invariant, giving%
\begin{equation}
\int_{K^{m}}\left[  (\nabla^{\ell_{1}}+\nabla^{\ell_{2}}+\nabla^{\ell_{3}%
}+\nabla^{\ell_{4}})(\nabla^{\ell_{1}}-\nabla^{\ell_{2}}+\nabla^{\ell_{3}%
}-\nabla^{\ell_{4}})g\right]  ~d\tilde{\mu}=0. \label{integratedIdentity2}%
\end{equation}

If we expand out the product on the left-hand side of
(\ref{integratedIdentity2}), we find that products of derivatives on
cyclically adjacent variables (e.g., $\nabla^{\ell_{1}}\cdot\nabla^{\ell_{2}}$
or $\nabla^{\ell_{4}}\cdot\nabla^{\ell_{1}}$) cancel, while products of
derivatives on \textquotedblleft opposite\textquotedblright\ variables (i.e.,
$\nabla^{\ell_{1}}\cdot\nabla^{\ell_{3}}$ and $\nabla^{\ell_{2}}\cdot
\nabla^{\ell_{4}}$) combine. Thus, (\ref{integratedIdentity2}) is precisely
equivalent to the desired identity (\ref{integratedIdentity}).

\subsection{An example\label{example.sec}}

\begin{figure}[ptb]
\centering
\includegraphics[
height=2.1949in,
width=2.1949in
]{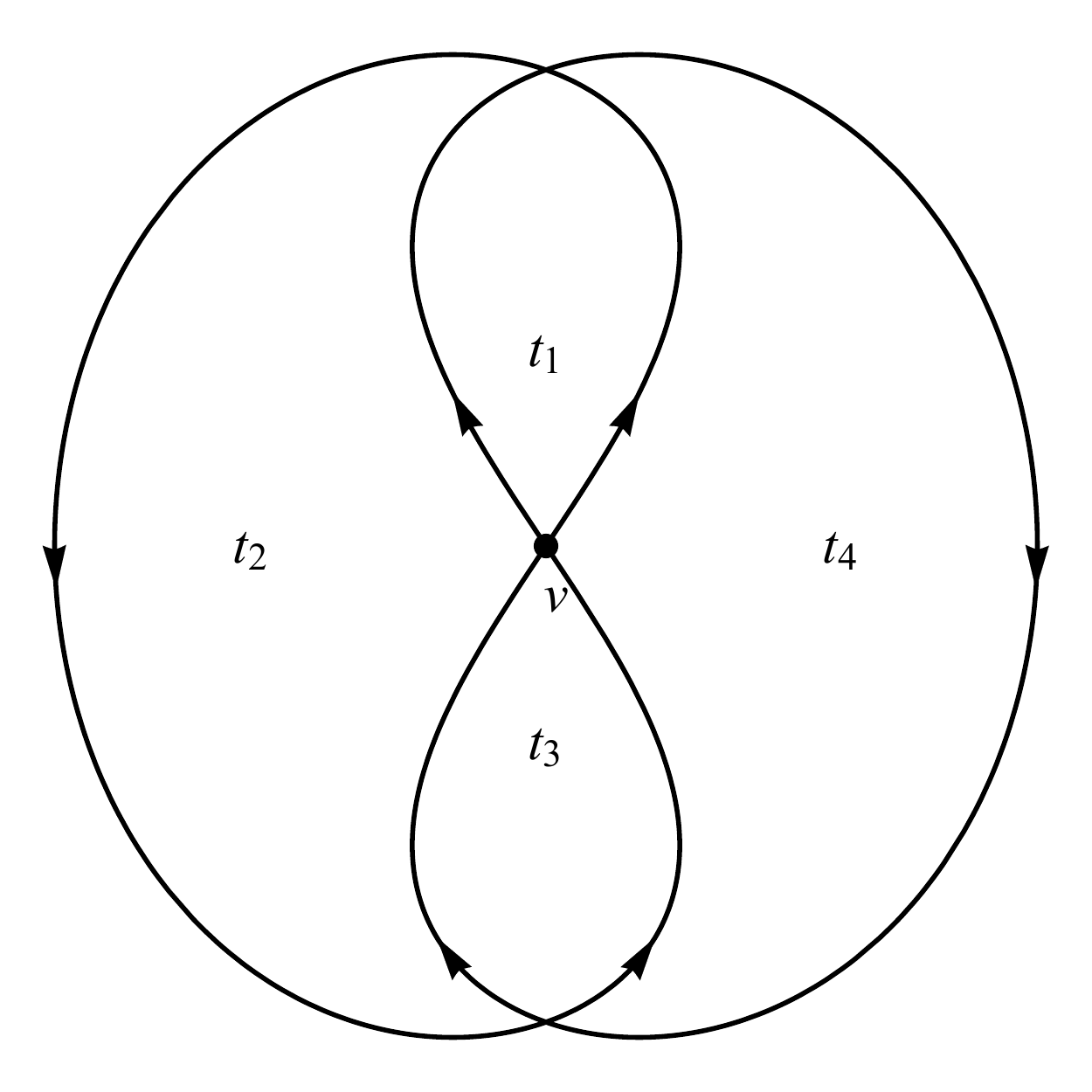}\caption{An example loop}%
\label{example.fig}%
\end{figure}

We now illustrate the preceding proof of the Makeenko--Migdal equation for the
loop in Figure \ref{example.fig}. We take the spanning tree $\mathbb{T}$ to
consisting of the edge between $F_{1}$ and $F_{2}$ and the edge between
$F_{3}$ and $F_{4}.$ In that case, it is easy to work out that the reduced
loops (with backtracks removed) associated to each bounded face $F$ of
$\mathbb{G}$ will simply proceed from $v$ around $F$ in the counter-clockwise
direction, as in Figure \ref{exampledecomp.fig}. It is straightforward to
check that the loop in Figure \ref{example.fig} decomposes as%
\[
(L_{1}L_{2}L_{3})(L_{1}^{-1}L_{4}^{-1}L_{3}^{-1}),
\]
where the notation means that we first traverse $L_{1},$ then $L_{2},$ and so
on. The expressions in parentheses indicate the component loops in the $U(N)$
version of the Makeenko--Migdal equation. (We will carry out the calculation
for a general compact group $K$ and then indicate what happens when $K=U(N).$)
Since parallel transport is order-reversing, the holonomy around $L$ is
expressed in terms of the loop variables as
\[
\mathrm{hol}(L)=\ell_{3}^{-1}\ell_{4}^{-1}\ell_{1}^{-1}\ell_{3}\ell_{2}%
\ell_{1}.
\]

\begin{figure}[ptb]
\centering
\includegraphics[
height=2.0185in,
width=2.1949in
]{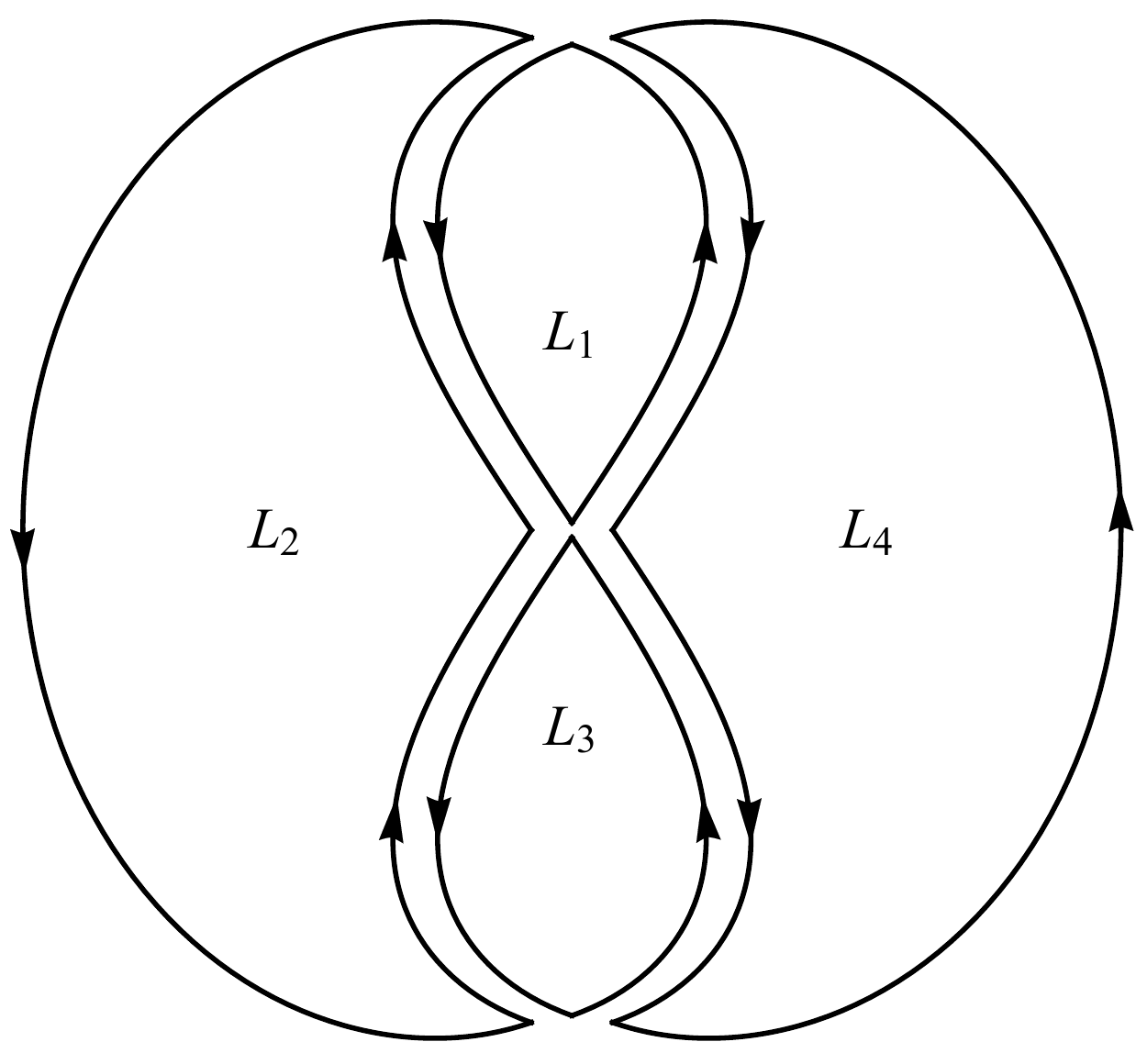}\caption{The generating loops for the example in Figure
\ref{example.fig}}%
\label{exampledecomp.fig}%
\end{figure}

We start by computing $\nabla^{a_{1}}\cdot\nabla^{a_{2}}g,$ with
\[
g(\ell_{1},\ell_{2},\ell_{3},\ell_{4})=\mathrm{tr}(\ell_{3}^{-1}\ell_{4}%
^{-1}\ell_{1}^{-1}\ell_{3}\ell_{2}\ell_{1}),
\]
with $\mathrm{tr}$ denoting the normalized trace in some representation of
$K.$ Recalling (\ref{adjacentLoops}), we find that the substitutions
$a_{1}\mapsto a_{1}e^{sX}$ and $a_{2}\mapsto a_{2}e^{tX}$ in the edge
variables translates into the substitutions
\[
(\ell_{1},\ell_{2},\ell_{3},\ell_{4})\mapsto(e^{-tX}\ell_{1}e^{sX},\ell
_{2}e^{tX},\ell_{3},e^{-sX}\ell_{4})
\]
in the loop variables. Thus $-\nabla^{a_{1}}\cdot\nabla^{a_{2}}g$ is equal to
\begin{align}
&  -\sum_{X}\left.  \frac{\partial^{2}}{\partial s~\partial t}\mathrm{tr}%
[\ell_{3}^{-1}(\ell_{4}^{-1}e^{sX})(e^{-sX}\ell_{1}^{-1}e^{tX})\ell_{3}%
(\ell_{2}e^{tX})(e^{-tX}\ell_{1}e^{sX})]\right\vert _{s=t=0}\nonumber\\
=  &  -\sum_{X}\left.  \frac{\partial^{2}}{\partial s~\partial t}%
\mathrm{tr}[\ell_{3}^{-1}\ell_{4}^{-1}\ell_{1}^{-1}e^{tX}\ell_{3}\ell_{2}%
\ell_{1}e^{sX}]\right\vert _{s=t=0}\nonumber\\
=  &  -\sum_{X}\mathrm{tr}[\ell_{3}^{-1}\ell_{4}^{-1}\ell_{1}^{-1}X\ell
_{3}\ell_{2}\ell_{1}X], \label{gradABexample}%
\end{align}
where in each term we sum $X$ over an orthonormal basis of the Lie algebra
$\mathfrak{k}$ of $K.$ In the $K=U(N)$ case with the normalized trace taken in
the standard representation, the last line of (\ref{gradABexample}) simplifies
to $\mathrm{tr}(\ell_{3}^{-1}\ell_{4}^{-1}\ell_{1}^{-1})\mathrm{tr}(\ell
_{3}\ell_{2}\ell_{1}),$ by (\ref{magic}).

We now compute the alternating sum of time derivatives of $\int g\,d\tilde
{\mu},$ using (\ref{LapTerm}). By Lemma \ref{adInv.lem}, we are free to
evaluate the Laplacians using any combination of derivatives on the left and
on the right. The computations work out most simply if we compute each
Laplacian as a product of a gradient on the left and a gradient on the right.
With this convention, we easily obtain
\begin{align*}
\mathrm{\Delta}_{\ell_{1}}g  &  =\sum_{X}\big(\mathrm{tr}(\ell_{3}^{-1}%
\ell_{4}^{-1}X\ell_{1}^{-1}X\ell_{3}\ell_{2}\ell_{1})-\mathrm{tr}(\ell
_{3}^{-1}\ell_{4}^{-1}X\ell_{1}^{-1}\ell_{3}\ell_{2}X\ell_{1})\\
&  -\mathrm{tr}(\ell_{3}^{-1}\ell_{4}^{-1}\ell_{1}^{-1}X\ell_{3}\ell_{2}%
\ell_{1}X)+\mathrm{tr}(\ell_{3}^{-1}\ell_{4}^{-1}\ell_{1}^{-1}\ell_{3}\ell
_{2}X\ell_{1}X)\big)\\
& \\
\mathrm{\Delta}_{\ell_{2}}g  &  =\sum_{X}\mathrm{tr}(\ell_{3}^{-1}\ell
_{4}^{-1}\ell_{1}^{-1}\ell_{3}X\ell_{2}X\ell_{1})
\end{align*}
\begin{align*}
\mathrm{\Delta}_{\ell_{3}}g  &  =\sum_{X}\big(\mathrm{tr}(X\ell_{3}^{-1}%
X\ell_{4}^{-1}\ell_{1}^{-1}\ell_{3}\ell_{2}\ell_{1})-\mathrm{tr}(X\ell
_{3}^{-1}\ell_{4}^{-1}\ell_{1}^{-1}X\ell_{3}\ell_{2}\ell_{1})\\
&  -\mathrm{tr}(\ell_{3}^{-1}X\ell_{4}^{-1}\ell_{1}^{-1}\ell_{3}X\ell_{2}%
\ell_{1})+\mathrm{tr}(\ell_{3}^{-1}\ell_{4}^{-1}\ell_{1}^{-1}X\ell_{3}%
X\ell_{2}\ell_{1})\big)
\end{align*}
\begin{align*}
\mathrm{\Delta}_{\ell_{4}}g  &  =\sum_{X}\mathrm{tr}(\ell_{3}^{-1}X\ell
_{4}^{-1}X\ell_{1}^{-1}\ell_{3}\ell_{2}\ell_{1}),
\end{align*}
where in each term, we sum $X$ over an orthonormal basis for $\mathfrak{k}.$

After taking half the alternating sum of these Laplacians, we obtain two
\textquotedblleft good\textquotedblright\ terms (the third term in
$\mathrm{\Delta}_{\ell_{1}}$ and the second term in $\mathrm{\Delta}_{\ell
_{3}}$), namely%
\begin{equation}
-\frac{1}{2}\mathrm{tr}(\ell_{3}^{-1}\ell_{4}^{-1}\ell_{1}^{-1}X\ell_{3}%
\ell_{2}\ell_{1}X)-\frac{1}{2}\mathrm{tr}(X\ell_{3}^{-1}\ell_{4}^{-1}\ell
_{1}^{-1}X\ell_{3}\ell_{2}\ell_{1}). \label{goodTerms}%
\end{equation}
After using the cyclic invariance of the trace, these terms reduce precisely
to (\ref{gradABexample}). We are left with eight \textquotedblleft
bad\textquotedblright\ terms in the alternating sum that must cancel out after integration.

To verify this cancellation directly, we compute that%
\begin{align}
&  (\nabla^{\ell_{1},L}\cdot\nabla^{\ell_{2},L}-\nabla^{\ell_{1},R}\cdot
\nabla^{\ell_{2},R})g\nonumber\\
&  =\sum_{X}\big(-\mathrm{tr}(\ell_{3}^{-1}\ell_{4}^{-1}X\ell_{1}^{-1}\ell
_{3}\ell_{2}X\ell_{1})+\mathrm{tr}(\ell_{3}^{-1}\ell_{4}^{-1}\ell_{1}^{-1}%
\ell_{3}\ell_{2}X\ell_{1}X)\nonumber\\
&  +\mathrm{tr}(\ell_{3}^{-1}\ell_{4}^{-1}\ell_{1}^{-1}X\ell_{3}X\ell_{2}%
\ell_{1})-\mathrm{tr}(\ell_{3}^{-1}\ell_{4}^{-1}\ell_{1}^{-1}\ell_{3}X\ell
_{2}X\ell_{1})\big) \label{Ad1}%
\end{align}
and that
\begin{align}
&  (\nabla^{\ell_{3},L}\cdot\nabla^{\ell_{4},L}-\nabla^{\ell_{3},R}\cdot
\nabla^{\ell_{4},R})g\nonumber\\
&  =\sum_{X}\big(\mathrm{tr}(X\ell_{3}^{-1}X\ell_{4}^{-1}\ell_{1}^{-1}\ell
_{3}\ell_{2}\ell_{1})-\mathrm{tr}(\ell_{3}^{-1}X\ell_{4}^{-1}\ell_{1}^{-1}%
\ell_{3}X\ell_{2}\ell_{1})\nonumber\\
&  -\mathrm{tr}(\ell_{3}^{-1}X\ell_{4}^{-1}X\ell_{1}^{-1}\ell_{3}\ell_{2}%
\ell_{1})+\mathrm{tr}(\ell_{3}^{-1}\ell_{4}^{-1}X\ell_{1}^{-1}X\ell_{3}%
\ell_{2}\ell_{1})\big). \label{Ad2}%
\end{align}
One can easily check that the eight bad terms in the alternating sum of
Laplacians are exactly the sum of the right-hand sides of (\ref{Ad1}) and
(\ref{Ad2}). Thus, by Lemma \ref{adInv.lem}, the bad terms integrate to zero.

\section{Reduction to the generic case\label{generic.sec}}

In the preceding sections, we assumed that our graph $\mathbb{G}$ was
\textquotedblleft generic\textquotedblright\ relative to the given vertex $v,
$ meaning that the four adjacent faces are distinct and bounded and that the
four edges emanating from $v$ are distinct. (More precisely, nongeneric
behavior of the edges is when one of the outgoing edges $e_{i}$ emanating from
$v$ coincides with $e_{j}^{-1}$ for some $j\neq i$.) In this section, we show
that a version of Theorem \ref{MMab.thm} still holds even if the preceding
assumptions are not satisfied.

If a graph $\mathbb{G}$ is not generic relative to a given vertex $v,$ we
construct a new graph $\mathbb{G}^{\prime}$ by adding four new vertices and
connecting them in a circular pattern as in Figures \ref{generic1.fig}. Our
strategy will be to \textquotedblleft promote\textquotedblright\ a function of
the edge variables of $\mathbb{G}$ to a function of the edge variables of
$\mathbb{G}^{\prime},$ apply the Makeenko--Migdal equation for $\mathbb{G}%
^{\prime},$ and then deduce the Makeenko--Migdal equation for $\mathbb{G}.$

Note that if $\mathbb{G}^{\prime}$ is as in Figure \ref{generic1.fig}, it will
satisfy the additional assumption used in Section \ref{loop.sec}, namely that
we can choose a spanning tree $\mathbb{T}^{\prime}$ for $\mathbb{G}^{\prime}$
so that the loops associated to the four faces adjacent to $v$ have the form
given in (\ref{adjacentLoopAssumption}). After all, if we choose
$\mathbb{T}^{\prime}$ to include all four edges coming out of $v,$ then for
$i=1,\ldots,4,$ there will only be one edge of $F_{i}$ that bounds $F_{i}$
positively and is not in $\mathbb{T}^{\prime},$ namely the \textquotedblleft
circular\textquotedblright\ edge (oriented counter-clockwise). Thus, the loop
given by L\'{e}vy's construction will simply proceed out along $e_{i},$ then
counter-clockwise around the circular edge, then back to $v$ along
$e_{i+1}^{-1}.$ Such a loop is of the form (\ref{adjacentLoopAssumption}).

\subsection{Consistency of the Yang--Mills measure\label{consistency.sec}}

The first key point is to establish a consistency result, stating that the
integral of the promoted function with respect to the Yang--Mills measure for
$\mathbb{G}^{\prime}$ is the same as the integral of the original function
with respect to the Yang--Mills measure for $\mathbb{G}.$

\begin{figure}[ptb]
\centering
\includegraphics[
height=1.5843in,
width=3.2785in
]{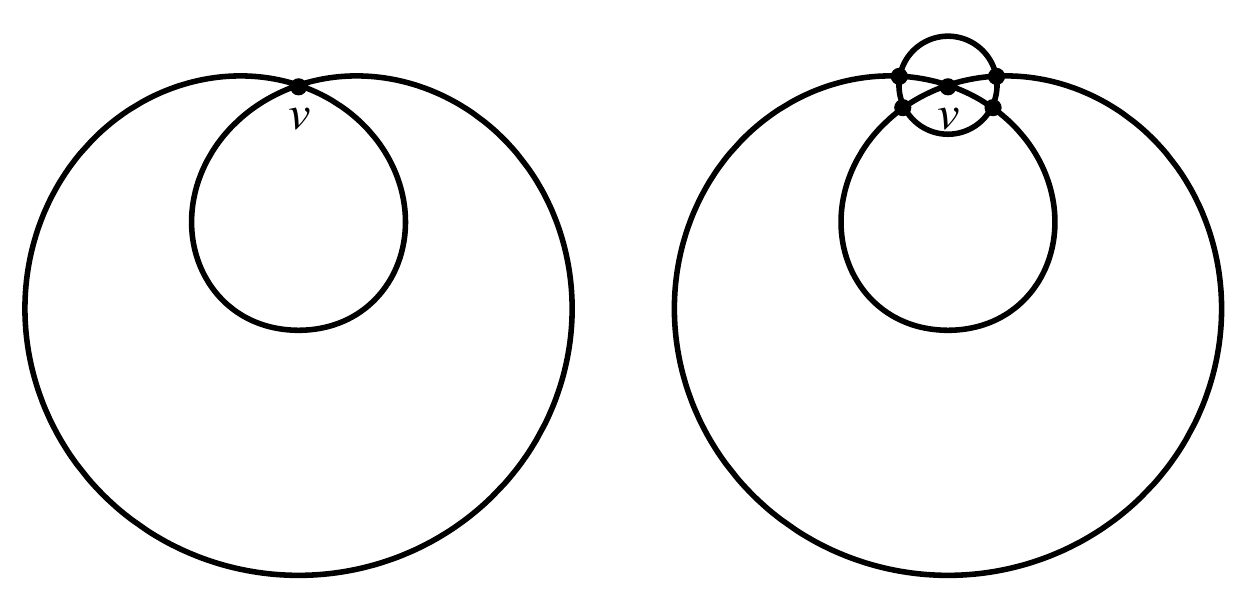}\caption{A nongeneric graph relative to $v$ (left) and its
generic counterpart (right)}%
\label{generic1.fig}%
\end{figure}

Conceptually, this consistency result holds (at least for gauge-invariant
functions) because both integrals are computing the same functional of the
underlying white noise in the path-integral formulation of the Yang--Mills
theory. It is also possible to establish consistency directly from the
formulas for the integrals over the graphs, as we now explain. The consistency
result has two aspects. First, suppose we add an extra vertex in the middle of
an edge $e,$ thus subdividing $e$ into two new edges $e^{\prime}$ and
$e^{\prime\prime}.$ Thus, $e$ is replaced by $e^{\prime}e^{\prime\prime}.$
Since parallel transport is order reversing, we will then replace the edge
variable $x$ associated with $e$ by the product $x^{\prime\prime}x^{\prime}$
of edge variables $x^{\prime}$ and $x^{\prime\prime}$ associated to
$e^{\prime}$ and $e^{\prime\prime}.$ Thus, if $f$ is a function of the edge
variables of the original graph, we form a function $f^{\prime}$ of the edge
variables of the new graph by setting%
\begin{equation}
f^{\prime}(x^{\prime},x^{\prime\prime},\mathbf{y})=f(x^{\prime\prime}%
x^{\prime},\mathbf{y}), \label{promotion1}%
\end{equation}
where $\mathbf{y}$ is the collection of all edge variables different from $x$
(in the original graph) or different from $x^{\prime}$ and $x^{\prime\prime}$
(in the new graph).

Consistency of the integrals under this change is easy to establish. If $\nu$
is the density of the Yang--Mills measure, it is easy to see from
(\ref{DriversForm}) that $\nu_{\mathbb{G}^{\prime}}(x^{\prime},x^{\prime
\prime},\mathbf{y})=\nu_{\mathbb{G}}(x^{\prime\prime}x^{\prime},\mathbf{y}).$
Thus,
\begin{align*}
\int f^{\prime}~d\mu_{\mathbb{G}^{\prime}}  &  =\int\int\int f(x^{\prime
\prime}x^{\prime},\mathbf{y})\nu_{\mathbb{G}}(x^{\prime\prime}x^{\prime
},\mathbf{y})~dx^{\prime}~dx^{\prime\prime}~d\mathbf{y}\\
&  =\int\int f(z,\mathbf{y})\nu_{\mathbb{G}}(z,\mathbf{y})~dz~d\mathbf{y}\\
&  =\int f~d\mu_{\mathbb{G}},
\end{align*}
where in the second equality, we have made the change of variable
$z=x^{\prime\prime}x^{\prime}$ in the $x^{\prime}$-integral and used the
normalization of the Haar measure in the $x^{\prime\prime}$ integral.

We also consider consistency under another type of change in the graph.
Suppose we create a new graph $\mathbb{G}^{\prime}$ from $\mathbb{G}$ by
keeping the vertex set the same and adding some new edges. Then any function
$f$ of the edge variables of $\mathbb{G}$ can be promoted to a function of the
edge variables of $\mathbb{G}^{\prime}$ by making the $f^{\prime}$ independent
of the new edge variables; that is,
\begin{equation}
f^{\prime}(\mathbf{x},\mathbf{y})=f(\mathbf{y}), \label{promotion2}%
\end{equation}
where $\mathbf{x}$ represents the edge variables for the new edges and
$\mathbf{y}$ represents the edges variables for the old edges. The consistency
identity%
\begin{equation}
\int f^{\prime}~d\mu_{\mathbb{G}^{\prime}}=\int f~d\mu_{\mathbb{G}}
\label{graphConsistency2}%
\end{equation}
is a special case of Theorem 1.22 in \cite{Levy1}. Note that adding an edge
can divide a face with area $t$ into two faces with areas $s$ and $s^{\prime}
$ satisfying $s+s^{\prime}=t.$ The key idea in verifying
(\ref{graphConsistency2}) is the convolution identity for heat kernels:
$\rho_{s}\ast\rho_{s^{\prime}}=\rho_{t}.$

\subsection{Interpretation of the theorem\label{interpret.sec}}

We now work toward establishing a version of the abstract Makeenko--Migdal
equation (Theorem \ref{MMab.thm}) in the nongeneric case. We must first
describe the proper interpretation of the theorem in the nongeneric case. If
one of the adjacent faces $F_{i}$ is the unbounded face, the corresponding
time derivative $\partial/\partial t_{i}$ should be interpreted as the zero
operator. If $F_{i}=F_{j}$ for $i\neq j,$ we simply have the same
area-derivative twice on the left-hand side of the Makeenko--Migdal equation.

\begin{figure}[ptb]
\centering
\includegraphics[
height=2.2in,
width=2.2in
]{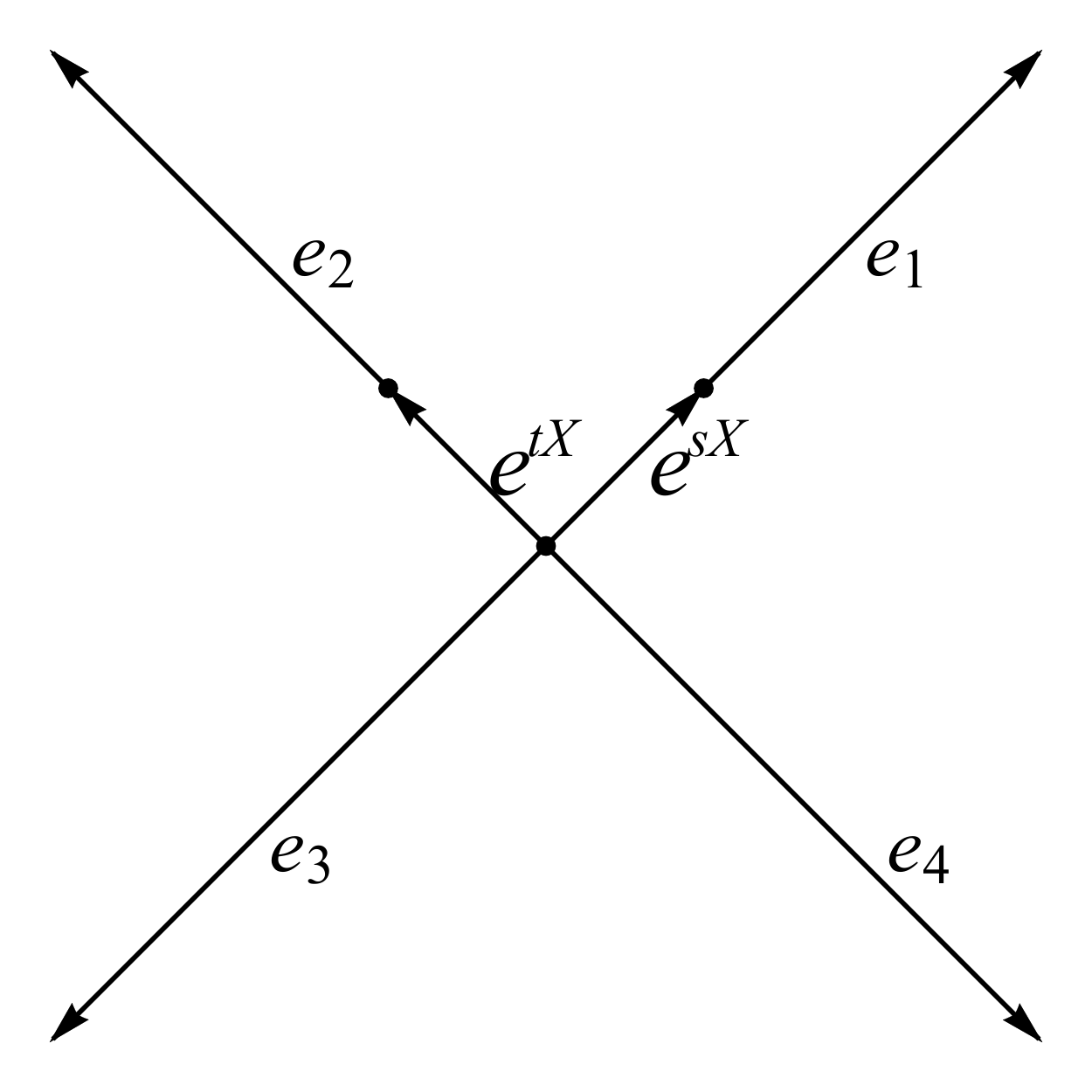}\caption{Geometric interpretation of $\nabla^{a_{1}}\cdot
\nabla^{a_{2}}f$}%
\label{gradab.fig}%
\end{figure}

Meanwhile, if for some $i\neq j,$ the edges $e_{i}$ and $e_{j}$ are inverses
of each other, we choose one of the indices (say, $i$) and we then no longer
view $a_{j}$ as an independent variable, but as simply another name for
$a_{i}^{-1}.$ Then to correctly interpret the expression $\nabla^{a_{1}}%
\cdot\nabla^{a_{2}}f,$ we insert a factor of $e^{sX}$ at the beginning of
$e_{1}$ and a factor of $e^{tX}$ at the beginning of $e_{2},$ as in Figure
\ref{gradab.fig}. We then determine the corresponding changes of the
independent variables---keeping in mind that parallel transport is order
reversing---differentiate at $s=t=0,$ and sum $X$ over an orthonormal basis.
Suppose, for example, that $e_{1}$ and $e_{2}$ are inverses of each other, but
$e_{3}$ and $e_{4}$ are distinct, and suppose we take $a_{1},$ $a_{3},$ and
$a_{4}$ as our independent variables. Then Figure \ref{gradab.fig} tells us
that we should replace $a_{1}$ by $e^{-tX}a_{1}e^{sX},$ so that
\[
(\nabla^{a_{1}}\cdot\nabla^{a_{2}}f)(a_{1},a_{3},a_{4},\mathbf{b})=\sum
_{X}\left.  \frac{\partial^{2}}{\partial s\partial t}f(e^{-tX}a_{1}%
e^{sX},a_{3},a_{4},\mathbf{b})\right\vert _{s=t=0}.
\]

Finally, we describe the correct notion of extended gauge invariance at $v$ in
the case where the edges $e_{1},\ldots,e_{4}$ are not necessarily distinct. We
insert two extra factors of $x$ as on either of the two sides of Figure
\ref{extended.fig}. A function $f$ has extended gauge invariance if the value
of $f$ is unchanged by this insertion. We may be more precise about this
definition as follows. Let $\mathbb{G}$ be a graph and $v$ a vertex of
$\mathbb{G}$ with four edges, where we count an edge $e$ twice if both ends of
$e$ are attached to $v.$ Let $f$ be a function of the edge variables of
$\mathbb{G}.$ Let $\mathbb{G}^{\prime}$ be either of the two graphs in Figure
\ref{extended.fig}, obtained by adding two new vertices to $\mathbb{G}.$ Let
$f^{\prime}(x,y,\mathbf{z})$ be the function of the edge variables of
$\mathbb{G}^{\prime}$ formed by the method described earlier in this section,
where $x$ and $y$ denote the edge variables associated to the two new edges
emanating from $v$ and $\mathbf{z}$ represents all the other edge variables.
We say that $f$ has extended gauge invariance if%
\[
f^{\prime}(x,x,\mathbf{z})=f^{\prime}(\mathrm{id},\mathrm{id},\mathbf{z}),
\]
for all $x\in K,$ where $\mathrm{id}$ is the identity element of $K.$

\begin{figure}[ptb]
\centering
\includegraphics[
height=1.6743in,
width=3.6391in
]{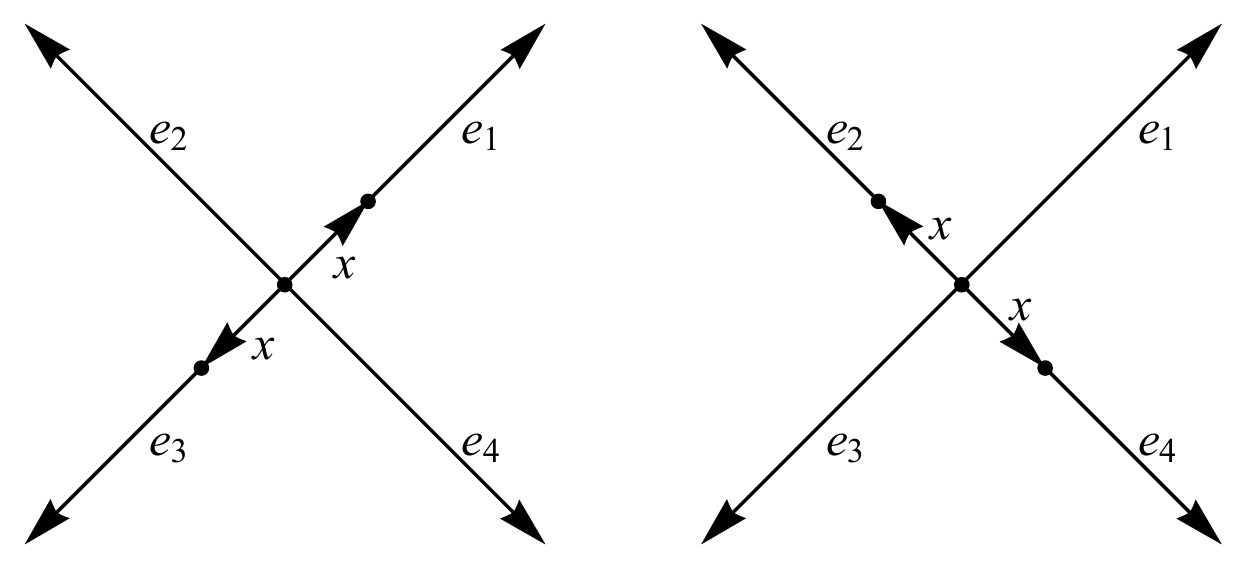}\caption{Geometric interpretation of extended gauge invariance}%
\label{extended.fig}%
\end{figure}

If the edges $e_{1},\ldots,e_{4}$ are distinct, this new notion of extended
gauge invariance agrees with the notion in Definition \ref{extended.def}.
After all, when the edges are distinct, $f^{\prime}(x,x,\mathbf{z})$ will
simply be (in the notation of Definition \ref{extended.def}) either
$f(a_{1}x,a_{2},a_{3}x,a_{4},\mathbf{b})$ or $f(a_{1},a_{2}x,a_{3}%
,a_{4}x,\mathbf{b}).$

We consider two other examples of our new notion of extended invariance.
Suppose that $e_{3}$ coincides with $e_{2}^{-1}$ (but $e_{1}$ and $e_{4}$ are
distinct), and that we choose $a_{1},$ $a_{2},$ and $a_{4}$ as our independent
variables. Since parallel transport is order reversing, the transformation
indicated by the left side of Figure \ref{extended.fig} will change $a_{2}$ to
$x^{-1}a_{2}$ and $a_{1}$ to $a_{1}x,$ while leaving $a_{4} $ unchanged.
Meanwhile, the transformation for the right-hand side of the figure changes
$a_{2}$ and $a_{4}$ to $a_{2}x$ and $a_{4}x$, as usual. Thus, when
$e_{3}=e_{2}^{-1},$ extended gauge invariance means that%
\[
f(a_{1},a_{2},a_{4},\mathbf{b})=f(a_{1}x,x^{-1}a_{2},a_{4},\mathbf{b}%
)=f(a_{1},a_{2}x,a_{4}x,\mathbf{b}),
\]
for all $x\in K.$ Similarly, suppose $e_{1}$ and $e_{3}$ are inverses of each
other (but $e_{2}$ and $e_{4}$ are distinct) and we take $a_{1},$ $a_{2},$ and
$a_{4}$ as our independent variables. Then extended gauge invariance means
that%
\[
f(a_{1},a_{2},a_{4},\mathbf{b})=f(x^{-1}a_{1}x,a_{2},a_{4},\mathbf{b}%
)=f(a_{1},a_{2}x,a_{4}x,\mathbf{b}),
\]
for all $x\in K.$

\subsection{Proof of the theorem}

If $\mathbb{G}$ is a graph that is not generic at $v,$ we consider another
graph $\mathbb{G}^{\prime}$ as in Figure \ref{generic1.fig}, which is generic
at $v.$ We label the edges and faces of $\mathbb{G}^{\prime}$ around $v$ as in
Figure \ref{generic2.fig}, and we let $a_{i}^{\prime}$ denote the edge
variable associated to $e_{i}^{\prime}.$ We now promote any function $f$ of
the edge variables of $\mathbb{G}$ to a function $f^{\prime}$ of the edge
variables of $\mathbb{G}^{\prime},$ by the method described in Section
\ref{consistency.sec}. Using the geometric interpretation of extended gauge
invariance in Figure \ref{extended.fig}, we can verify that if $f$ has
extended gauge invariance at $v$, so does $f^{\prime}.$ Furthermore, if $f$ is
gauge invariant--- as we assume in our loop-based proof in Section
\ref{loop.sec}---it is not hard to see that $f^{\prime}$ is also gauge invariant.

\begin{figure}[ptb]
\centering
\includegraphics[
height=1.9251in,
width=3.6391in
]{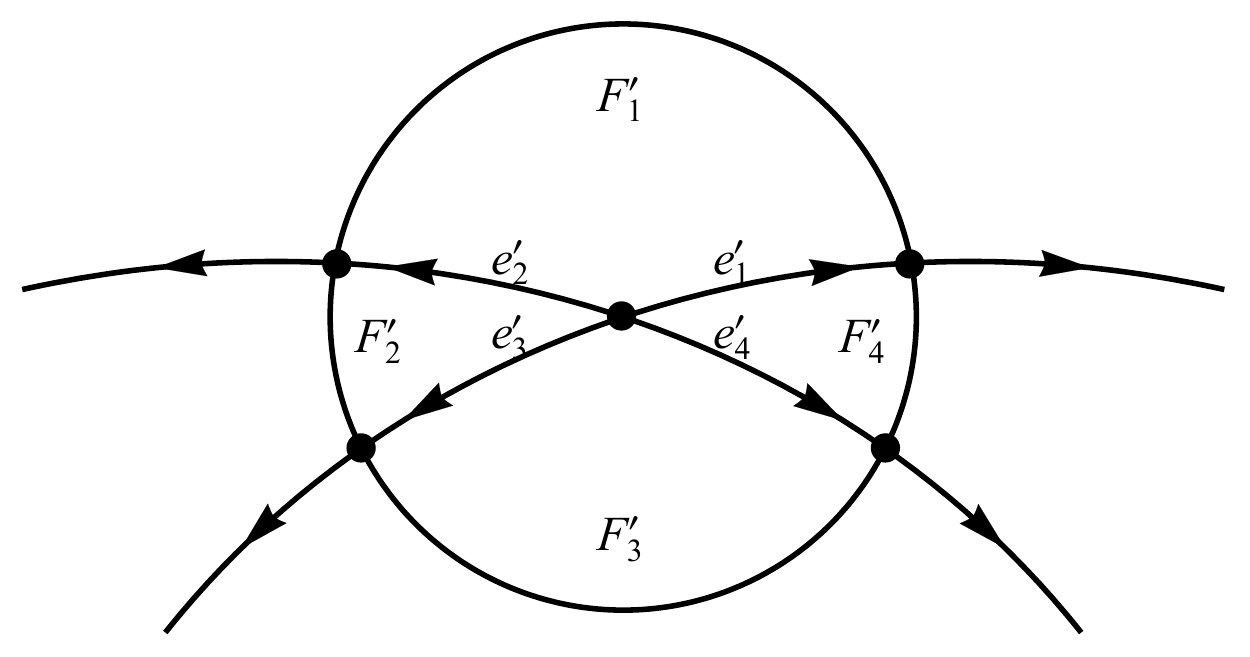}\caption{The adjacent faces and edge variables for the generic
graph}%
\label{generic2.fig}%
\end{figure}

Since the new graph is generic relative to $v$, all of our proofs of the
Makeenko--Migdal equation apply to $\int f^{\prime}~d\mu_{\mathbb{G}^{\prime}%
}$, where in the case of the loop-based proof, we would need to assume that
$f$ (and thus, also, $f^{\prime}$) is gauge invariant. (As we have noted in
the introduction to this section, $\mathbb{G}^{\prime}$ has the additional
property, used in Section \ref{loop.sec}, that we can choose a spanning tree
for $\mathbb{G}^{\prime}$ so that the adjacent loops have the form in
(\ref{adjacentLoopAssumption}).) It now remains only to see that
the\ Makeenko--Migdal equation for $\int f^{\prime}~d\mu_{\mathbb{G}^{\prime}%
}$ reduces to the Makeenko--Migdal equation for $\int f~d\mu_{\mathbb{G}}.$

We begin with the time derivatives and we consider first the possibility that
one of the adjacent faces $F_{i}$ in the original graph is the unbounded face.
In that case, the face $F_{i}^{\prime}$ in the new graph will share a
\textquotedblleft circular\textquotedblright\ edge with the unbounded face.
Since the circular edge lies between $F_{i}^{\prime}$ and the unbounded face,
the corresponding edge variable $c_{i}$ will occur in \textit{only one} of the
heat kernels in the definition of $\mu_{\mathbb{G}^{\prime}}.$ Thus, the
density of $\mu_{\mathbb{G}^{\prime}}$ will take the form%
\[
\rho_{\tilde{t}_{i}}(c_{i}\gamma)\delta,
\]
where $\gamma$ is a word in edge variables other than $c_{i}$ and where
$\delta$ is a product of heat kernels evaluated on edge variables other than
$c_{i}.$ Thus, if $t_{i}^{\prime}$ denotes the area of $F_{i}^{\prime},$ we
have%
\begin{align*}
\frac{\partial}{\partial t_{i}^{\prime}}\int f^{\prime}~d\mu_{\mathbb{G}%
^{\prime}}  &  =\frac{1}{2}\int f^{\prime}~\mathrm{\Delta}_{c_{i}}%
[\rho_{\tilde{t}_{i}}(c_{i}\gamma)]\delta~d\mathrm{Haar}\\
&  =\frac{1}{2}\int(\mathrm{\Delta}_{c_{i}}f^{\prime})~d\mu_{\mathbb{G}%
^{\prime}}\\
&  =0,
\end{align*}
since $f^{\prime}$ is, by construction, independent of $c_{i}.$

We next consider the possibility that for some $i\neq j,$ two bounded faces
$F_{i}$ and $F_{j}$ in the original graph coincide. In that case, the face
$F_{i}=F_{j}$ is divided into three faces in the new graph, $F_{i}^{\prime}$,
$F_{j}^{\prime},$ and one other face $G.$ Thus,%
\[
t_{i}=t_{i}^{\prime}+t_{j}^{\prime}+s,
\]
where $s$ is the area of $G,$ which means that varying $t_{i}^{\prime}$ has
the same effect as varying $t_{i}.$ It follows from this observation and the
consistency of the Yang--Mills measure that%
\begin{equation}
\frac{\partial}{\partial t_{i}^{\prime}}\int f^{\prime}~d\mu_{\mathbb{G}%
^{\prime}}=\frac{\partial}{\partial t_{i}}\int f~d\mu_{\mathbb{G}}.
\label{titi}%
\end{equation}
If three or more bounded faces in the original graph coincide, a very similar
argument shows that (\ref{titi}) still holds.

Finally, using the geometric interpretation of $\nabla^{a_{1}}\cdot
\nabla^{a_{2}}$ in Figure \ref{gradab.fig}, we can verify that
\begin{equation}
(\nabla^{a_{1}}\cdot\nabla^{a_{2}}f)^{\prime}=\nabla^{\tilde{a}_{1}}%
\cdot\nabla^{\tilde{a}_{2}}f^{\prime}, \label{aiCommute}%
\end{equation}
showing that the right-hand side of the Makeenko--Migdal equation for
$\mathbb{G}^{\prime}$ reduces to the corresponding expression for $\mathbb{G}.
$

\begin{acknowledgement}
We thank Franck Gabriel for comments on and corrections to an earlier
version of this article. We also thank Ambar Sengupta for many useful
discussions about Yang--Mills theory and the master field.  Finally, we thank
the referees for a careful reading of the manuscript, which has led to several
clarifications and corrections and a significant improvement to the paper.
\end{acknowledgement}


\begin{thebibliography}{999999999}                                                                                        %


\bibitem[AS]{AS}M. Anshelevich and A. N. Sengupta, Quantum free Yang-Mills on
the plane. \textit{J. Geom. Phys.} \textbf{62} (2012), 330--343.

\bibitem[Chatt]{Chatt}S. Chatterjee, Rigorous solution of strongly coupled
$SO(N)$ lattice gauge theory in the large $N$ limit, preprint: arXiv:1502.07719.

\bibitem[Dahl]{Dahl}A. Dahlqvist, Free energies and fluctuations for the
unitary Brownian motion. \textit{Comm. Math. Phys.} \textbf{348}, no. 2 (2016), 395-444.

\bibitem[DM]{DM}H. G. Dosch and V. F. M\"{u}ller, Lattice gauge theory in two
spacetime dimensions, \textit{Fortschritte der Physik} \textbf{27} (1979), 547--559.

\bibitem[Dr]{Dr}B. K. Driver, YM$_{2}$: continuum expectations, lattice
convergence, and lassos. \textit{Comm. Math. Phys.} \textbf{123} (1989), 575--616.

\bibitem[DHK]{DHK}B. K. Driver, B. C. Hall, and T. Kemp, The large-$N$ limit
of the Segal--Bargmann transform on $\mathbb{U}_{N}$, \textit{J. Funct. Anal.}
\textbf{265} (2013), 2585--2644.

\bibitem[DGHK]{DHKsurf}B. K. Driver, F. Gabriel, B. C. Hall, and T. Kemp, The
Makeenko--Migdal equation for Yang--Mills theory on compact surfaces,
preprint: arXiv:1602.03905.

\bibitem[Gop]{Gop}R. Gopakumar, The master field in generalised QCD$_{2},$
\textit{Nuclear Phys. B} \textbf{471} (1996), 246--260.

\bibitem[GG]{GG}R. Gopakumar and D. Gross, Mastering the master field,
\textit{Nuclear Phys. B} \textbf{451} (1995), 379--415.

\bibitem[GKS]{GKS}L. Gross, C. King, and A. N. Sengupta, Two-dimensional
Yang-Mills theory via stochastic differential equations, \textit{Ann. Physics}
\textbf{194} (1989), 65--112.

\bibitem[K]{KazakovU(N)}A. Kazakov, Wilson loop average for an arbitary
contour in two-dimensional $U(N)$ gauge theory, \textit{Nuclear Physics}
\textbf{B179} (1981) 283-292.

\bibitem[KK]{KK}V. A. Kazakov and I. K. Kostov, Non-linear strings in
two-dimensional $U(\infty)$ gauge theory, \textit{Nucl. Phys. B} \textbf{176}
(1980), 199-215.

\bibitem[L\'{e}vy1]{Levy1}T. L\'{e}vy, Two-dimensional Markovian holonomy
fields. Ast\'{e}risque No. 329 (2010), 172 pp.

\bibitem[L\'{e}vy2]{Levy}T. L\'{e}vy, The master field on the plane, preprint: arXiv:1112.2452.

\bibitem[MM]{MM}Y. M. Makeenko and A. A. Migdal, Exact equation for the loop
average in multicolor QCD, \textit{Physics Letters} \textbf{88B} (1979), 135-137.

\bibitem[MO]{MO}P. Menotti and E. Onofri, The action of $SU(N)$ lattice gauge
theory in terms of the heat kernel on the group manifold, \textit{Nuclear
Physics B} \textbf{190} (1981), 288-300

\bibitem[Mig]{Migdal}A. A. Migdal, Recursion equations in gauge field
theories, \textit{Sov. Phys. JETP} \textbf{42} (1975), 413-418.

\bibitem[Sing]{Sing}I. M. Singer, On the master field in two dimensions,
\textit{In}: Functional analysis on the eve of the 21st century, Vol. 1 (New
Brunswick, NJ, 1993), 263--281, Progr. Math., 131, Birkh\"{a}user Boston,
Boston, MA, 1995.

\bibitem[Sen]{Sengupta1}A. N. Sengupta, Traces in two-dimensional QCD: the
large-$N$ limit. \textit{Traces in number theory, geometry and quantum
fields}, 193--212, Aspects Math., E38, Friedr. Vieweg, Wiesbaden, 2008.

\bibitem['t Hooft]{tHooft}G. 't Hooft, A planar diagram theory for strong
interactions, \textit{Nuclear Physics B }\textbf{72} (1974), 461-473.
\end{thebibliography}
\end{document}